\documentclass{article}
\usepackage{amsfonts,amssymb,amsmath,a4wide,paralist}
\usepackage{epsfig,rotating,framed,enumerate,listings}
\usepackage{color}
\usepackage{paralist}
\usepackage{marginnote,nicefrac}
\usepackage{authblk}
\usepackage{tikz}
\usepackage{xfrac}

\usepackage[paper=a4paper,left=28mm,right=28mm,top=30mm, bottom=25mm]{geometry}

\usepackage{hyperref}

\newcommand{\bigo}{\ensuremath{\mathcal{O}}}
\newcommand{\dpws} {\text{d-pw}}

\newcommand{\pws} {\text{pw}}

\newcommand{\dvsns} {\text{d-vsn}}

\newcommand{\IN}{\mathbb{N}}

\newcommand{\ideg}{\text{indegree}}
\newcommand{\odeg}{\text{outdegree}}

\newcommand{\gansfuss}[1]{\mbox{``}{#1}\mbox{''}}

\newcommand{\un} {{\it und}}

\newcommand{\lcws} {\text{lcw}}

\newcommand{\vsns} {\mbox{vsn}}
\newcommand{\dnws} {\text{d-nw}}
\newcommand{\dlrws} {\text{d-lrw}}
\newcommand{\cutws} {\text{cutw}}
\newcommand{\dlcws} {\text{d-lcw}}
\newcommand{\lnlcws} {\text{lnlcw}}
\newcommand{\dcutws} {\text{d-cutw}}
\newcommand{\dlnlcws} {\text{d-lnlcw}}
\newcommand{\nws} {\text{nw}}
\newcommand{\lrws} {\text{lrw}}
\newcommand{\free} {\text{Free}}
\newcommand{\lab} {\text{lab}}
\newcommand{\val} {\text{val}}
\newcommand{\rg} {\text{rg}}
\newcommand{\dpwtwos} {\text{u-pw}}

\newtheorem{theorem}{Theorem}[section]
\newtheorem{example}[theorem]{Example}
\newtheorem{lemma}[theorem]{Lemma}
\newtheorem{definition}[theorem]{Definition}

\newtheorem{observation}[theorem]{Observation}
\newtheorem{corollary}[theorem]{Corollary}

\newtheorem{remark}[theorem]{Remark}

\definecolor{DGray}{gray}{0.6}
\definecolor{Gray}{gray}{0.7}
\definecolor{XGray}{gray}{0.45}
\definecolor{Gray2}{gray}{0.8}
\definecolor{LGray}{gray}{0.90}
\definecolor{BGray}{gray}{0.00}

\newenvironment{proof}{\noindent{\bf Proof~}}{\null\hfill $\Box$\par\medskip}

\definecolor{light-gray}{gray}{0.5}

\begin{document}

\title{Comparing Linear Width Parameters for Directed Graphs}

\author{Frank Gurski}
\author{Carolin Rehs}

\affil{\small University of  D\"usseldorf,
Institute of Computer Science, Algorithmics for Hard Problems Group,\newline 
40225 D\"usseldorf, Germany}

\maketitle

\begin{abstract}
In this paper we introduce the linear clique-width, 
linear NLC-width, neighbourhood-width, and 
linear rank-width for directed graphs. We compare these
parameters with each other as well as with the previously 
defined parameters directed path-width 
and directed cut-width. It turns out that the parameters directed linear clique-width, 
directed linear NLC-width, directed neighbourhood-width, and directed
linear rank-width are equivalent in that sense, that all of these parameters
can be upper bounded by each of the others. For the restriction 
to digraphs of bounded vertex degree  directed path-width 
and directed cut-width are equivalent. Further for the restriction to
semicomplete digraphs of bounded vertex degree all six mentioned 
width parameters are equivalent.
We also show close relations of the measures to their undirected versions
of the underlying undirected graphs, which allow us to show
the hardness of computing the considered linear width parameters for directed graphs. 
Further we give first characterizations for directed graphs defined by parameters of small width.

\end{abstract}

%%%%%%%%%%%%%%%%%%%%%%%%%%%%%%%%%%%%%%%%%%%%%%%%%%%%%%%%%%%%%%%%%%%%%%%%%%
\section{Introduction}
%%%%%%%%%%%%%%%%%%%%%%%%%%%%%%%%%%%%%%%%%%%%%%%%%%%%%%%%%%%%%%%%%%%%%%%%%%

A {\em graph parameter} is a function that associates with every graph a positive integer.
Examples for graph parameters are tree-width \cite{RS86}, clique-width \cite{CO00},
NLC-width \cite{Wan94}, and rank-width \cite{OS06}. Clique-width, NLC-width,
and rank-width are equivalent, i.e. a graph has bounded clique-width, if and only if
it has bounded NLC-width and that is if and only if it has bounded rank-width. 
The latter three parameters are more general than tree-width, since
graphs of bounded tree-width also have bounded clique-width but even for
dense graphs (e.g. cliques) the tree-width is unbounded while the
clique-width can be small \cite{GW00}.
Graph classes of bounded width are interesting from an algorithmic point of view 
since several hard graph problems can be solved in 
polynomial time  by dynamic programming along the tree structure of
the input graph, see \cite{Arn85,AP89,Hag00,KZN00} and \cite{CE12,CMR00,EGW01a}.
Furthermore such parameters are also interesting from a structural point of view,
e.g. in the research of special graph classes \cite{BDLM05,Bod98}.

In this paper we consider graph parameters which are defined by the existence of an 
underlying path-structure for the input graph. These are
path-width \cite{RS83}, 
cut-width \cite{AH73},
linear clique-width \cite{GW05a}, 
linear NLC-width \cite{GW05a},
neighbourhood-width \cite{Gur06a}, and 
linear rank-width \cite{Gan11}. 
With the exception of cut-width these parameters can be regarded as restrictions of 
the above mentioned parameters with underlying tree-structure to an underlying path-structure. 
The relation between these parameters corresponds to their tree-structural
counterparts, since 
%bounded cut-width implies bounded path-width and 
bounded path-width
implies bounded linear NLC-width, linear clique-width, neighborhood-width, and linear rank-width.
Further the reverse direction is not true in general, see \cite{Gur06a}.
Such restrictions to underlying path-structures  are often helpful to show results for the
general parameters, see \cite{FRRS09,FGLSZ18}.
These linear parameters are also interesting from a structural point of view,
e.g. in the research of special graph classes \cite{Gan11,Gur06,HMP11}.

% directed

Since several problems and applications frequently use directed graphs,
during the last years, width parameters for directed graphs 
have received a lot of attention, see \cite{GHKLOR14,GHKMORS16} and
the two book chapters \cite[Chapter 9]{BG18} and \cite[Chapter 6]{DE14}.
Lifting the above mentioned parameters using an underlying tree-structure to directed graphs lead to
directed tree-width \cite{JRST01}, directed NLC-width \cite{GWY16}, directed clique-width \cite{CO00}, 
and directed rank-width \cite{KR13}.

% linear directed

In this paper we study directed  graph parameters which are defined by the existence of an 
underlying path-structure for the input graph. One of the most famous examples is the
directed path-width, which was introduced by Reed, Seymour, and Thomas around 1995 
(see \cite{Bar06}) and  studied in \cite{Bar06,Tam11,KKT15,KKKTT16}. Further
the cut-width for directed graphs was introduced by  Chudnovsky et al. in \cite{CFS12}.
Regarding the usefulness of linear width parameters for undirected graphs we introduce 
the directed linear NLC-width, directed linear clique-width,
directed neighborhood-width, and directed linear rank-width. 
In contrast to the linear width measures for undirected graphs, 
for directed graphs their relations
turn out to be more involved. Table \ref{tab-fam} shows some classes of digraphs 
demonstrating various possible combinations of the listed width measures being bounded and
unbounded.

\begin{table}[ht]
\begin{center}
{\tiny
\begin{tabular}{|l|lr|lr|c|c|c|c|c|}
\hline
                     & \multicolumn{2}{c|}{undirected} &    \multicolumn{2}{c|}{directed} & DAG   & CB  & BS &OP  & TT \\
\hline
cut-width            & $\cutws$  & \cite{AH73}   & $\dcutws$  & \cite{CFS12}     &  0       &  $\infty$ &  $\infty$ &0 & 0  \\
path-width           & $\pws$    & \cite{RS83}   & $\dpws$    & Thomas et al.    &  0       &  $\infty$ & 1 &0 & 0      \\
linear clique-width  & $\lcws$   & \cite{GW05a}  & $\dlcws$   & here             & $\infty$ & 2   & 2  &3 & 2 \\ 
linear NLC-width     & $\lnlcws$ & \cite{GW05a}  & $\dlnlcws$ & here             & $\infty$ & 1 & 1 &3 & 1 \\
neighbourhood-width  & $\nws$    & \cite{Gur06a} & $\dnws$    & here             & $\infty$ & 1 & 1 &2 & 1 \\
linear rank-width    & $\lrws$   & \cite{Gan11}  & $\dlrws$   & here             & $\infty$ & 1 & 1 &2 & 1 \\
\hline
\end{tabular}
}
\end{center}\label{tab-fam}
\caption{Width measures and their values for 
directed acyclic digraphs (DAG), complete bioriented (CB) digraphs, bioriented stars (BS),
oriented paths (OP), and transitive
tournaments (TT).}
\end{table}

For all these  linear width parameters for directed graphs
we compare the directed width of a digraph
and the undirected width of its  underlying undirected graph, which allow us to show
the hardness of computing the considered linear width parameters for directed graphs. 

In order to  classify graph parameters we call two graph parameters $\alpha$ and $\beta$
{\em equivalent}, if there are two functions $f_1$ and $f_2$ such that for every digraph
$G$ the value $\alpha(G)$ can be upper bounded by $f_1(\beta(G))$ and the 
value $\beta(G)$ can be upper bounded by $f_2(\alpha(G))$. If $f_1$ and $f_2$ are
polynomials or linear functions, we call $\alpha$ and $\beta$ {\em polynomially equivalent}
or {\em linearly equivalent}, respectively.
We show that for general digraphs we have three sets of pairwise equivalent 
parameters, namely $\{\dcutws\}$, $\{\dpws\}$, and $\{\dnws,\dlnlcws,\dlcws,\dlrws\}$.
For  digraphs of bounded vertex degree this reduces to two sets 
$\{\dcutws,\dpws\}$ and $\{\dnws,\dlnlcws,\dlcws,\dlrws\}$ and
for semicomplete digraphs of bounded vertex degree
all these six  graph parameters are pairwise equivalent. 
With the exception of directed rank-width, the same results 
are even shown for polynomially and linearly equivalence.

By introducing  the class of directed threshold graphs, we
give characterizations for graphs defined by parameters of small width.

%Finally in Section \ref{sec-concl} we give conclusions and open problems.

%%%%%%%%%%%%%%%%%%%%%%%%%%%%%%%%%%%%%%%%%%%%%%%%%%%%%%%%%%%%%%%%%%%%%%%%%%
\section{Preliminaries}\label{sec-pre}
%%%%%%%%%%%%%%%%%%%%%%%%%%%%%%%%%%%%%%%%%%%%%%%%%%%%%%%%%%%%%%%%%%%%%%%%%%

We use the notations of Bang-Jensen and Gutin \cite{BG09} for graphs and digraphs.

\subsection{Undirected graphs}
We work with finite undirected {\em graphs} $G=(V,E)$,
where $V$ is a finite set of {\em vertices} 
and $E \subseteq \{ \{u,v\} \mid u,v \in
V,~u \not= v\}$ is a finite set of {\em edges}.
For a vertex $v\in V$ we denote by $N_G(v)$ 
the set of all vertices which are adjacent to $v$ in $G$, 
i.e.~$N_G(v)=\{w\in V~|~\{v,w\}\in E\}$. 
Set $N_G(v)$ is called the set of all {\em neighbors} of $v$ 
in $G$ or {\em neighborhood} of $v$ in $G$.  
The {\em degree} of a vertex $v\in V$,  denoted by $\deg_G(v)$, 
is the number of neighbors of vertex $v$ in $G$, i.e.~$\deg_G(v)=|N_G(v)|$. 
The maximum vertex degree is $\Delta(G)=\max_{v\in V} \deg_G(v)$.
A graph $G'=(V',E')$ is a {\em subgraph} of graph $G=(V,E)$ if $V'\subseteq V$ 
and $E'\subseteq E$.  If every edge of $E$ with both end vertices in $V'$  is in
$E'$, we say that $G'$ is an {\em induced subgraph} of digraph $G$ and 
we write $G'=G[V']$.
For some graph class $F$ we define 
$\free(F)$ as the set of all graphs $G$ such that no induced subgraph of $G$ is isomorphic
to a member of  $F$.

%For some undirected graph $G=(V,E)$ its complement graph is defined by
%$\co G=(V,\{\{u,v\}~|~\{u,v\}\not\in E, u,v\in V, u\neq v\})$.

\paragraph{Special Undirected Graphs} 
We recall some special graphs.
By $$P_n=(\{v_1,\ldots,v_n\},\{\{v_1,v_2\},\ldots, \{v_{n-1},v_n\}\}),$$ $n \ge 2$,
we denote a path on $n$ vertices and by  $$C_n=(\{v_1,\ldots,v_n\},\{\{v_1,v_2\},\ldots, \{v_{n-1},v_n\},\{v_n,v_1\}\}),$$
$n \ge 3$, we denote a cycle on $n$ vertices.
Further by $$K_n=(\{v_1,\ldots,v_n\},\{\{v_i,v_j\}~|~1\leq i<j\leq n\}),$$
$n \ge 1$, we denote a complete graph on $n$ vertices and by $$K_{n,m}=(\{v_1,\ldots,v_n,w_1,\ldots,w_m\},
\{\{v_i,w_j\}~|~ 1\leq i\leq n, 1\leq j \leq m\})$$  a complete bipartite graph on $n+m$ vertices.

\subsection{Directed graphs}
A {\em directed graph} or {\em digraph} is a pair  $G=(V,E)$, where $V$ is 
a finite set of {\em vertices} and  $E\subseteq \{(u,v) \mid u,v \in
V,~u \not= v\}$ is a finite set of ordered pairs of distinct\footnote{Thus we do not consider 
directed graphs with loops.} vertices called {\em arcs}. 
For a vertex $v\in V$, the sets $N_G^+(v)=\{u\in V~|~ (v,u)\in E\}$ and 
$N_G^-(v)=\{u\in V~|~ (u,v)\in E\}$ are called the {\em set of all out-neighbours}
and the {\em set of all in-neighbours} of $v$. 
The  {\em outdegree} of $v$, $\odeg_G(v)$ for short, is the number
of out-neighbours of $v$ and the  {\em indegree} of $v$, $\ideg_G(v)$ for short, 
is the number of in-neighbours of $v$ in $G$.
The {\em maximum out-degree} is $\Delta^+(G)=\max_{v\in V} \odeg_G(v)$ and
the {\em maximum in-degree} is $\Delta^-(G)=\max_{v\in V} \ideg_G(v)$.
The {\em maximum vertex degree} is $\Delta(G)=\max_{v\in V} \odeg_G(v)+\ideg_G(v)$. 
%and the  {\em maximum semi-degree} is $\Delta^0(G)=\max\{\Delta^-(G),\Delta^+(G)\}$.
%A vertex  $v\in V$ is  {\em out-dominating (in-dominated)} if
%it is adjacent to every other vertex in $V$ and is a source (a sink, respectively).
A digraph $G'=(V',E')$ is a {\em subdigraph} of digraph $G=(V,E)$ if $V'\subseteq V$ 
and $E'\subseteq E$.  If every arc of $E$ with both end vertices in $V'$  is in
$E'$, we say that $G'$ is an {\em induced subdigraph} of digraph $G$ and we 
write $G'=G[V']$. For some digraph class $F$ we define 
$\free(F)$ as the set of all digraphs $G$ such that no induced subdigraph of $G$ is isomorphic
to a member of  $F$.

%For some digraph $G=(V,E)$ its {\em complement digraph} is defined by
%$\co G=(V,\{(u,v)~|~(u,v)\not\in E, u,v\in V, u\neq v\})$.
%and its {\em converse digraph} is defined by
%$\con G=(V,\{(u,v)~|~(v,u)\in E, u,v\in V, u\neq v\})$.

Let $G=(V,E)$ be a digraph.
\begin{itemize}
\item $G$ is {\em edgeless} if for all $u,v \in V$, $u \neq v$, 
none of the two pairs $(u,v)$ and $(v,u)$ belongs to $E$.

\item $G$ is a {\em tournament} if for all $u,v \in V$, $u \neq v$, 
exactly one of the two pairs $(u,v)$ and $(v,u)$ belongs to $E$.

\item $G$ is {\em semicomplete} if for all $u,v \in V$, $u \neq v$, 
at least one of the two pairs $(u,v)$ and $(v,u)$ belongs to $E$.

\item $G$ is {\em (bidirectional) complete} if for all $u,v \in V$, $u \neq v$, 
both of the two pairs $(u,v)$ and $(v,u)$ belong to $E$.
\end{itemize}

\paragraph{Omitting the directions} 
For some given digraph $G=(V,E)$, we define
its {\em underlying undirected graph} by ignoring the directions of the 
edges, i.e.~$\un(G)=(V,\{ \{u,v\} \mid (u,v) \in E \text{ or } (v,u) \in E\})$.

\paragraph{Orientations} 
There are several ways to define a digraph $G=(V,E)$ from an undirected 
graph $G_u=(V,E_u)$.
If we replace every edge $\{u,v\}\in E_u$ by 
\begin{itemize}
\item
one of the arcs $(u,v)$ and $(v,u)$, we denote $G$ as an {\em orientation} of $G_u$.
Every digraph $G$  which can be obtained by an orientation of some undirected
graph $G_u$ is called an {\em oriented graph}.

\item
one or both of the arcs $(u,v)$ and $(v,u)$, we denote $G$ as a {\em biorientation} of $G_u$.
Every digraph $G$  which can be obtained by a biorientation of some undirected
graph $G_u$ is called a {\em bioriented graph}.

\item
both arcs $(u,v)$ and $(v,u)$, we denote $G$ as a {\em complete biorientation} of $G_u$.
Since in this case $G$ is well defined by $G_u$ we also denote
it by $\overleftrightarrow{G_u}$.
Every digraph $G$  which can be obtained by a complete biorientation of some undirected
graph $G_u$ is called a {\em complete bioriented graph}. 
\end{itemize}

\paragraph{Special directed graphs}
We recall some special directed graphs.
%We denote by $\overleftrightarrow{K_n}=(\{v_1,\ldots,v_n\},\{ (v_i,v_j)~|~1\leq i\neq j\leq n\})$,
%$n \ge 1$ a complete directed graph on $n$ vertices.
By $$\overrightarrow{P_n}=(\{v_1,\ldots,v_n\},\{ (v_1,v_2),\ldots, (v_{n-1},v_n)\}),$$ $n \ge 2$
we denote a directed path on $n$ vertices and by  
$$\overrightarrow{C_n}=(\{v_1,\ldots,v_n\},\{(v_1,v_2),\ldots, (v_{n-1},v_n),(v_n,v_1)\}),$$ $n \ge 2$
we denote a directed cycle on $n$ vertices. Further let 
$$\overleftrightarrow{K_n}=(\{v_1,\ldots,v_n\},\{ (v_i,v_j)~|~1\leq i\neq j\leq n\})$$
be a bidirectional complete digraph on $n$ vertices.
The {\em $k$-power graph} $G^k$ of a digraph $G$ is a graph with the 
same vertex set as $G$. There is an arc $(u,v)$ in $G^k$ if and only if there 
is a directed path from $u$ to $v$ of length at most $k$ in $G$. 
An {\em oriented forest (tree)} is the orientation of a forest (tree). A digraph is
an {\em out-tree (in-tree)} if it is an oriented tree in which there is
exactly one vertex of indegree (outdegree) zero. 
A {\em directed acyclic digraph (DAG for short)} is a digraph without any $\overrightarrow{C_n}$,
$n\geq 2$ as subdigraph.

%%%%%%%%%%%%%%%%%%%%%%%%%%%%%%%%%%%%%%%%%%%%%%%%%%%%%%%%%%%%%%%%%%%%%%%%%%
\section{Linear width parameters for directed graphs}\label{sec-pre-para}
%%%%%%%%%%%%%%%%%%%%%%%%%%%%%%%%%%%%%%%%%%%%%%%%%%%%%%%%%%%%%%%%%%%%%%%%%%

A {\em layout} of a graph $G=(V,E)$ is a bijective function $\varphi:V \to \{1,\ldots,|V|\}$.  
For a graph $G$, we denote by $\Phi(G)$ the set of all layouts for $G$.
Given a layout $\varphi\in \Phi(G)$ we define for $1\leq i\leq |V|$ 
the vertex sets \[L(i,\varphi,G)=\{u\in V ~|~ \varphi(u)\leq i \}  \text{ and } 
R(i,\varphi,G)=\{u\in V ~|~ \varphi(u) > i \}.\]
The {\em reverse layout} $\varphi^R$, for $\varphi\in \Phi(G)$, 
is defined by $\varphi^R(u)=|V|-\varphi(u)+1$, $u\in V$.

%%%%%%%%%%%%%%%%%%%%%%%%%%%%%%%%%%%%%%%%%%%%%%%%%%%%%%%%%%%%%%%%%%%%%%%%%%
\subsection{Directed path-width}\label{sec-dpw}
%%%%%%%%%%%%%%%%%%%%%%%%%%%%%%%%%%%%%%%%%%%%%%%%%%%%%%%%%%%%%%%%%%%%%%%%%%

%According to Bar{\'a}t \cite{Bar06}, 

The path-width ($\pws$) for undirected graphs 
was introduced in \cite{RS83}. 
The notion of directed path-width was
introduced by Reed, Seymour, and Thomas around 1995 (cf.\  \cite{Bar06}) and relates to directed
tree-width introduced by Johnson, Robertson, Seymour, and Thomas in
\cite{JRST01}.\footnote{Please note that there are some works which define the path-width of a digraph $G$ in
a different and not equivalent way by using the
path-width of  $\un(G)$, see Section  \ref{sec-concl}.
}

\begin{definition}[directed path-width]\label{def-dpw}
Let $G=(V,E)$ be a digraph.
A {\em directed path-decom\-po\-sition} of $G$
is a sequence $(X_1, \ldots, X_r)$ of subsets of $V$, called {\em bags}, such 
that the following three conditions hold true.
\begin{enumerate}[(1)]
\item $X_1 \cup \ldots \cup X_r ~=~ V$, 
\item\label{def-dpw-2} for each $(u,v) \in E$ there is a pair $i \leq j$ such that
  $u \in X_i$ and $v \in X_j$, and 
\item for all $i,j,\ell$ with $1 \leq i < j < \ell \leq r$ it holds
  $X_i \cap X_\ell \subseteq X_j$. 
\end{enumerate}
The {\em width} of a directed path-decomposition ${\cal X}=(X_1, \ldots, X_r)$ 
is $$\max_{1 \leq i \leq r} |X_i|-1.$$ The {\em directed path-width} of $G$,
$\dpws(G)$ for short, is 
the smallest integer $w$ such that there is a directed path-de\-com\-po\-sition for 
$G$ of width $w$. 
\end{definition}

%A  directed path-decomposition $(X_1, \ldots, X_r)$ is {\em nice}, if 
%$X_1=\emptyset$, $X_r=\emptyset$, and $|X_i-X_{i-1}|+|X_{i-1}-X_i|=1$ 
%for every $2\leq i \leq r$. If $X_i=X_{i-1}\cup\{v\}$, we denote
%$X_i$ as an {\em introduce node} and if  $X_{i-1}=X_{i}\cup\{v\}$, we denote
%$X_i$ as a {\em forget node}. Every directed path-decomposition can
%be transformed into a nice directed path-decomposition in time
%$\bigo(|V|^2)$. By the definition of a path-decomposition, within a
%nice directed path-decomposition  every graph vertex
%is introduced and forgotten exactly one. Thus every nice 
%path-decomposition has $r=2|V|+1$ bags.

There are a number of results on algorithms for computing directed path-width.
The directed path-width of a digraph $G=(V,E)$ can be computed in time 
$\bigo(\nicefrac{|E|\cdot |V|^{2\dpws(G)}}{(\dpws(G)-1)!})$ 
by \cite{KKKTT16} and in time 
$\bigo(\dpws(G)\cdot|E|\cdot |V|^{2\dpws(G)})$ 
by \cite{Nag12}. This leads to XP-algorithms
for directed path-width w.r.t.~the standard parameter
and implies that for each constant $w$, it is decidable in polynomial time whether a given 
digraph has directed path-width at most $w$. 
Further it is shown in \cite{KKT15} how to decide whether the directed
path-width of an $\ell$-semicomplete digraph is at most $w$ 
in time $(\ell+2w+1)^{2w}\cdot n^{\bigo(1)}$.
Furthermore the directed path-width can be computed in time 
$3^{\tau(\un(G))} \cdot |V|^{\bigo(1)}$, where  $\tau(\un(G))$ denotes the vertex
cover number of the underlying undirected graph of $G$, by \cite{Kob15}.
For sequence digraphs with a
given decomposition into $k$ sequence the directed path-width 
can be computed  in time $\bigo(k\cdot (1+N)^k)$, where $N$ denotes the
maximum sequence length \cite{GRR18}.
Further the directed path-width (and also the directed tree-width)
can be computed in linear time for  directed co-graphs \cite{GR18c}.

%The next lemma follows by the definition of converse digraphs and
%path-decompositions.
%
%\begin{lemma}\label{le-pw} 
%Let $G$ be a digraph.
%Sequence $(X_1, \ldots, X_r)$ is a directed path-decomposition 
%for $G$ if and only if sequence $(X_r, \ldots, X_1)$ a directed path-decomposition 
%of $\con G$. 
%\end{lemma}
%%
%
%
%
%\begin{lemma}\label{le-pw2} Let $G$ be a digraph.
%Let $G$ be some digraph, then $\dpws(G)=\dpws(\con G)$. 
%\end{lemma}

Example for digraphs of small directed path-width are given 
in Example \ref{ex-dvsn}, when considering the equivalent  (cf.\ Lemma \ref{th-pw-vs})
notation of directed vertex separation number.

%%%%%%%%%%%%%%%%%%%%%%%%%%%%%%%%%%%%%%%%%%%%%%%%%%%%%%%%%%%%%%%%%%%%%%%%%%%
\subsection{Directed vertex separation number}
%%%%%%%%%%%%%%%%%%%%%%%%%%%%%%%%%%%%%%%%%%%%%%%%%%%%%%%%%%%%%%%%%%%%%%%%%%%

The vertex separation number ($\vsns$) for undirected graphs 
was introduced in \cite{LT79}. 
In \cite{YC08} the directed vertex separation number for a digraph $G=(V,E)$ 
has been introduced as follows.

\begin{definition}[directed vertex separation number, \cite{YC08}]
The {\em directed vertex separation number} of a digraph $G=(V,E)$ is
defined as follows. 
\begin{equation}
\dvsns(G) = \min_{\varphi\in \Phi(G)} \max_{1\leq i\leq |V|} |\{u\in L(i,\varphi,G)~|~ \exists v \in R(i,\varphi,G): (v,u)\in E\}|\label{vsn-d1}
\end{equation}
\end{definition}

%For every optimal layout $\varphi$ we obtain the same value when we consider 
%the arcs forward in the reverse ordering $\varphi^R$. Thus we obtain an equivalent 
%definition as follows (cf.~\cite{BG09}).
%
%
%\begin{equation}
%\dvsns(G) = \min_{\varphi\in \Phi(G)} \max_{1\leq i\leq |V|} |\{u\in R(i,\varphi,G)~|~ \exists v \in L(i,\varphi,G): (v,u)\in E\}|\label{vsn-d2}
%\end{equation}

Since the converse digraph has the same path-width as its 
original graph, 
%by Lemma \ref{le-pw2}
we obtain an equivalent 
definition, which will be useful later on.
\begin{equation}
\dvsns(G) = \min_{\varphi\in \Phi(G)} \max_{1\leq i\leq |V|} |\{u\in L(i,\varphi,G)~|~ \exists v \in R(i,\varphi,G): (u,v)\in E\}|\label{vsn-d3a}
\end{equation}
%\begin{equation}
%\dvsns(G) = \min_{\varphi\in \Phi(G)} \max_{1\leq i\leq |V|} |\{u\in R(i,\varphi,G)~|~ \exists v \in L(i,\varphi,G): (u,v)\in A\}|\label{vsn-d4a}
%\end{equation}

%%

\begin{example}[directed vertex separation number]\label{ex-dvsn}
\label{ex-vs}
\begin{enumerate}[(1.)]
\item Every  directed path $\overrightarrow{P_n}$
has directed vertex separation number 0. 
%This can be shown by the layout $\varphi(v_i)=i$, $1\leq i \leq n$.

\item
The $k$-power graph $(\overrightarrow{P_n})^k$ of a directed path $\overrightarrow{P_n}$ has  directed vertex 
separation number 0.

\item  Every directed cycle
$\overrightarrow{C_n}$
has directed  vertex separation number 1. 
%This can be shown by the layout $\varphi(v_i)=i$, $1\leq i \leq n$.

\item The bidirectional complete digraph $\overleftrightarrow{K_3}$
and the complete biorientation of a star $K_{2,2,2}$ have directed vertex 
separation number $2$.\footnote{We use the complete biorientations of the two forbidden
minors for the set of all graphs of vertex separation number 1, see \cite[Fig.\ 1]{KL94}.}

%star_{2,2,2}

\item
Every bidirectional complete digraph $\overleftrightarrow{K_n}$ has 
directed vertex 
separation number $n-1$.

\end{enumerate}
\end{example}

%%%%%%%%%%%%%%%%%%%%%%%%%%%%%%%%%%%%%%%%%%%%%%%%%%%%%%%%%%%%%%%%%%%%%%%%%%
\subsection{Directed cut-width}
%%%%%%%%%%%%%%%%%%%%%%%%%%%%%%%%%%%%%%%%%%%%%%%%%%%%%%%%%%%%%%%%%%%%%%%%%%

The cut-width ($\cutws$) of undirected graphs was  introduced in \cite{AH73}.
The cut-width of digraphs was introduced by  Chudnovsky, Fradkin, 
and Seymour in \cite{CFS12}. 
%The directed cut-width of some 
%digraph $G=(V,E)$  is defined by an ordering of vertices similar to
%undirected cut-width, with the exception
%that only arcs directed forward in the ordering contribute to the width of a cut.
%
%

\begin{definition}[directed cut-width, \cite{CFS12}]
The {\em directed cut-width} of digraph $G=(V,E)$ is 
\begin{equation}
\dcutws(G) = \min_{\varphi\in \Phi(G)} \max_{1\leq i\leq |V|} | (u,v)\in E ~|~ u \in L(i,\varphi,G), v\in R(i,\varphi,G)\}|.\label{cutw-d1}
\end{equation}
\end{definition}

For every optimal layout $\varphi$ we obtain the same value when we consider 
the arcs backwards in the 
reverse ordering $\varphi^R$. Thus we obtain an equivalent definition, 
which will be useful later on.
\begin{equation}
\dcutws(G) = \min_{\varphi\in \Phi(G)} \max_{1\leq i\leq |V|} | (v,u)\in E ~|~ u \in L(i,\varphi,G), v\in R(i,\varphi,G)\}|\label{cutw-d2}
\end{equation}

Subexponential parameterized algorithms for computing the directed cut-width of semicomplete digraphs
are given in \cite{FP13a}.

\begin{example}[directed cut-width] 
\label{ex-cut-w}
\begin{enumerate}[(1.)]
\item Every directed path  $\overrightarrow{P_n}$ has directed cut-width 0. 
%This can be shown by the layout $\varphi(v_i)=n-i+1$, $1\leq i \leq n$.

\item 
The $k$-power graph $(\overrightarrow{P_n})^k$ of a directed path $\overrightarrow{P_n}$ 
has  directed directed cut-width $0$.

\item Every directed cycle  $\overrightarrow{C_n}$
has directed cut-width 1. 

\item The bidirectional complete digraph $\overleftrightarrow{K_3}$ has 
directed cut-width $2$.

\item 
Every bidirectional complete digraph $\overleftrightarrow{K_n}$ has 
directed cut-width $\lfloor \frac{n}{2} \rfloor\cdot \lceil \frac{n}{2} \rceil$.

\end{enumerate}
\end{example}

%%%%%%%%%%%%%%%%%%%%%%%%%%%%%%%%%%%%%%%%%%%%%%%%%%%%%%%%%%%%%%%%%%%%%%%%%%
\subsection{Directed linear NLC-width}
%%%%%%%%%%%%%%%%%%%%%%%%%%%%%%%%%%%%%%%%%%%%%%%%%%%%%%%%%%%%%%%%%%%%%%%%%%

The linear NLC-width ($\lnlcws$) for undirected graphs was introduced in \cite{GW05a} 
as a parameter
by restricting the NLC-width\footnote{The abbreviation NLC results from the {\em
node label controlled} embedding mechanism originally defined for graph
grammars.}, defined in \cite{Wan94}, 
to an underlying path-structure. Next we introduce the corresponding
parameter for directed graphs by  a modification of the edge 
inserting operation  $\times_S$ of the linear NLC-width, which also
leads to a restriction of directed NLC-width \cite{GWY16}.
Let $[k]=\{1,\ldots,k\}$ be the set of all integers between $1$ and $k$.

\begin{definition}[directed linear NLC-width]
The {\em directed linear NLC-width} of a digraph $G$, $\dlnlcws(G)$ for short, is
the minimum number of labels
needed to define $G$ using the following four operations:
\begin{enumerate}
\item Creation of a new vertex  with label $a$ (denoted by $\bullet_a$).
\item Disjoint union of a labeled digraph $G$ and a single vertex $v$ labeled by $a$
plus all arcs between label pairs from $\overrightarrow{S}$ directed from $G$ to $v$
and all  arcs between label pairs from $\overleftarrow{S}$ directed from $v$ to $G$
for two relations $\overrightarrow{S}$ and $\overleftarrow{S}$
(denoted by $G \otimes_{(\overrightarrow{S}, \overleftarrow{S})} \bullet_a$).

\item Change every label $a$ into label $R(a)$ by some function $R$ (denoted by $\circ_R$).
\end{enumerate}
\end{definition}

The {\em directed linear NLC-width} of an unlabeled digraph $G=(V,E)$ is 
the smallest integer $k$, such that there is a mapping  $\lab : V \to [k]$ 
such that the labeled digraph  $(V,E,\lab)$  has directed linear NLC-width  
at  most $k$.
An expression $X$ built with the operations
defined above is called a {\em directed linear NLC-width $k$-expression}.
Note that every expression defines 
a layout by the order in which the vertices are inserted in the corresponding digraph.
 The digraph defined by expression $X$ is denoted by $\val(X)$.
%Every such expression has by its recursive definition
%a tree structure which we call the {\em directed linear NLC-width expression tree}.

\begin{example}[directed linear NLC-width] 
\label{ex-d-lnlcw}
\begin{enumerate}[(1.)]
\item 
Every bidirectional complete digraph $\overleftrightarrow{K_n}$ has 
directed  linear NLC-width 1. 

\item The directed paths  $\overrightarrow{P_3}$
and $\overrightarrow{P_4}$ have directed  linear NLC-width  2.

\item 
Every directed path  $\overrightarrow{P_n}$
has directed  linear NLC-width at most 3. 
%This
%can  be shown by the following 3-expressions $X_{n}$.
%\[X_{3}= (\bullet_1  \otimes_{(\{(1,2)\},\emptyset)} \bullet_2) \otimes_{(\{(2,3)\},\emptyset)} \bullet_3\]
%\[X_{n}= \circ_{\{(1,1),(2,1),(3,2)\}}(X_{n-1}) \otimes_{(\{(2,3)\},\emptyset)}  \bullet_3, ~n\ge 4 \]

\item 
Every directed cycle  $\overrightarrow{C_n}$
has directed  linear NLC-width at most 4.

\item Every $k$-power graph  $(\overrightarrow{P_n})^k$  of a directed path $\overrightarrow{P_n}$
has directed linear NLC-width at most $k+2$. 
%This can
%be shown by an expression which inserts the vertices of $\overrightarrow{P_n}$
%in the order $v_1,\ldots,v_n$.

\item Every complete biorientation of a grid $\overleftrightarrow{G_n}$, $n\geq 3$,
has directed linear NLC-width   at least $n$ and at most $n+2$, see \cite{GR00,Gur08b}.

\end{enumerate}
\end{example}

% GR00 cw(G_n) >=n

% nw(G_n) <=n+1 my paper

%%%%%%%%%%%%%%%%%%%%%%%%%%%%%%%%%%%%%%%%%%%%%%%%%%%%%%%%%%%%%%%%%%%%%%%%%%
\subsection{Directed linear clique-width}
%%%%%%%%%%%%%%%%%%%%%%%%%%%%%%%%%%%%%%%%%%%%%%%%%%%%%%%%%%%%%%%%%%%%%%%%%%

The linear clique-width ($\lcws$) for undirected graphs was introduced in \cite{GW05a} 
as a parameter by restricting the clique-width, defined in \cite{CO00}, 
to an underlying path-structure. Next we introduce the corresponding
parameter for directed graphs by  a modification of the edge 
inserting operation of the linear clique-width, which also
leads to a restriction for directed clique-width \cite{CO00}.

\begin{definition}[directed linear clique-width]
The {\em directed linear clique-width} of a digraph $G$, $\dlcws(G)$ for short, is
the minimum number of labels
needed to define $G$ using the following four operations:
\begin{enumerate}
\item Creation of a new vertex  with label $a$ (denoted by $\bullet_a$).
\item Disjoint union of a labeled digraph $G$ and a single vertex labeled by $a$
(denoted by $G\oplus \bullet_a$).
\item Inserting an arc from every vertex with label $a$ to every vertex with label $b$
($a\neq b$, denoted by $\alpha_{a,b}$).
\item Change label $a$ into label $b$ (denoted by $\rho_{a\to b}$).
\end{enumerate}
\end{definition}

The {\em linear clique-width} of an unlabeled digraph $G=(V,E)$ is the smallest integer
$k$, such that there is a mapping  $\lab : V \to [k]$ such that
the labeled digraph  $(V,E,\lab)$  has linear linear clique-width at  most $k$. 
An expression $X$ built with the operations
defined above is called a 
{\em directed linear clique-width $k$-expression}. Note that every expression defines 
a layout by the order in which the vertices are inserted in the corresponding digraph.
%Every such expression has by its recursive definition
%a tree structure which we call the {\em directed linear clique-width expression tree}.
The digraph defined by expression $X$ is denoted by $\val(X)$.

\begin{example}[directed linear clique-width]\label{ex-dlcw}
\label{ex-d-lcw}
\begin{enumerate}[(1.)]
\item Every edgeless digraph has 
directed  linear clique-width 1.

\item
Every bidirectional complete digraph $\overleftrightarrow{K_n}$ has 
directed  linear clique-width 2.

\item Every directed path $\overrightarrow{P_n}$
has directed  linear clique-width at most 3. 
%This
%can  be shown by the following 3-expressions  $Y_{n}$.
%\[Y_{3}= \alpha_{2,3}(\alpha_{1,2}(\bullet_1 \oplus \bullet_2) \oplus \bullet_3)\]
%\[Y_{n}= \alpha_{2,3} (\rho_{3 \to 2}  (\rho_{2\to 1}  (Y_{n-1})) \oplus \bullet_3), ~n\ge 4 \]

\item 
Every directed cycle  $\overrightarrow{C_n}$
has directed  linear clique-width at most 4.

\item Every $k$-power graph $(\overrightarrow{P_n})^k$  of a directed path   $\overrightarrow{P_n}$
has directed linear clique-width at most $k+2$. 
%This can
%be shown by an expression which inserts the vertices of $\overrightarrow{P_n}$
%in the order $v_1,\ldots,v_n$. 
For $n\geq k(k+1)+2$  the given bound on the  directed linear clique-width  
is even exact by Corollary \ref{power-exact}.

\item Every complete biorientation of a grid $\overleftrightarrow{G_n}$, $n\geq 3$,
has directed linear clique-width   at least $n$ and at most $n+2$, see \cite{GR00,Gur08b}.

\end{enumerate}
\end{example}

% GR00 cw(G_n) >=n

% nw(G_n) <=n+1 my paper

%%%%%%%%%%%%%%%%%%%%%%%%%%%%%%%%%%%%%%%%%%%%%%%%%%%%%%%%%%%%%%%%%%%%%%%%%%
\subsection{Directed neighbourhood-width}
%%%%%%%%%%%%%%%%%%%%%%%%%%%%%%%%%%%%%%%%%%%%%%%%%%%%%%%%%%%%%%%%%%%%%%%%%%

The neighborhood-width ($\nws$) for undirected graphs was introduced in \cite{Gur06a}.
It differs from linear NLC-width and linear clique-width at most by one
but it is independent of vertex labels.

Let $G=(V,E)$ be a digraph and $U,W\subseteq V$ two disjoint vertex sets. The
set of all out-neighbours of $u$ into set $W$ and the set of all in-neighbours of $u$ into set $W$ 
are defined by $N^+_W(u)=\{v\in W ~|~(u,v)\in E\}$ and  $N^-_W(u)=\{v\in W ~|~ (v,u)\in E\}$.
The {\em directed neighbourhood} of vertex $u$ into set $W$ is defined by
$N_W(u)=(N^+_W(u),N^-_W(u))$ and  the set of all directed 
neighbourhoods of the vertices of set $U$ into set $W$ is
$N(U,W)=\{N_W(u) ~|~ u \in U\}$.
For some  layout $\varphi\in \Phi(G)$ we define 
$\dnws(\varphi,G) = \max_{1\leq i\leq |V|}|N(L(i,\varphi,G),R(i,\varphi,G))|$.

\begin{definition}[directed neighbourhood-width]
The {\em directed neighbour\-hood-width} of a digraph $G$ is 
\[\dnws(G) = \min_{\varphi\in \Phi(G)} \dnws(\varphi,G).\]
\end{definition}

\begin{example}[directed neighbourhood-width]\label{ex-nw} 
\begin{enumerate}[(1.)]
\item
Every bidirectional complete digraph $\overleftrightarrow{K_n}$ has 
directed  neighbourhood-width 1.

\item Every directed path  $\overrightarrow{P_n}$
has directed  neighbourhood-width at most 2. 
%This can  be shown by the layout $\varphi(v_i)=i$, $1\leq i \leq n$.

\item Every directed cycle  $\overrightarrow{C_n}$
has directed  neighbourhood-width at most 3. 
%This can  be shown by the layout $\varphi(v_i)=i$, $1\leq i \leq n$.

\item 
Every $k$-power graph  $(\overrightarrow{P_n})^k$ of a directed path 
 $\overrightarrow{P_n}$
has directed neigh\-bour\-hood-width at most $k+1$. 
%This
%can  be shown by the layout $\varphi(v_i)=i$, $1\leq i \leq n$.
For $n\geq k(k+1)+2$  the given bound on the  directed neigh\-bour\-hood-width  
is even exact by Corollary \ref{power-exact}.

\item Every complete biorientation of a grid $\overleftrightarrow{G_n}$, $n\geq 3$,
has directed  neigh\-bour\-hood-width at least $n$ and at most $n+1$, see \cite{GR00,Gur08b}.

\end{enumerate}
\end{example}

%%%%%%%%%%%%%%%%%%%%%%%%%%%%%%%%%%%%%%%%%%%%%%%%%%%%%%%%%%%%%%%%%%%%%%%%%%
\subsection{Directed linear rank-width}
%%%%%%%%%%%%%%%%%%%%%%%%%%%%%%%%%%%%%%%%%%%%%%%%%%%%%%%%%%%%%%%%%%%%%%%%%%

The rank-width for directed graphs was introduced in  Kant\'e in \cite{KR13}.
In \cite{Gan11} the linear rank-width  ($\lrws$)  for undirected graphs was introduced
by restricting the tree-structure of a rank decomposition to  caterpillars,
which is also possible for the directed case as follows.

Let  $G=(V,E)$ a digraph and $V_1,V_2\subset V$ be a disjoint
partition of the vertex set of $G$. Further let $M_{V_1}^{V_2}=(m_{ij})$ be the 
adjacent matrix defined over the four-element field
GF(4) for partition  $V_1\cup V_2$, i.e.
 \[
m_{ij}    = \left\{
 \begin{array}{lll}
   0 & \hspace{1cm}\mbox{if } (v_i,v_j)\not \in E  \mbox{ and }  (v_j,v_i)\not \in E  \\
 \mathbb{a}&  \hspace{1cm}\mbox{if }  (v_i,v_j)\in E      \mbox{ and }  (v_j,v_i)\not \in E  \\
  \mathbb{a}^2 &\hspace{1cm}\mbox{if }  (v_i,v_j)\not \in E      \mbox{ and }  (v_j,v_i) \in E  \\
 1 &\hspace{1cm}\mbox{if }  (v_i,v_j)\in E      \mbox{ and }  (v_j,v_i)\in E \\

     \end{array}
   \right.
\]

In GF(4) we have four elements $\{0,1,\mathbb{a},\mathbb{a}^2 \}$ with the properties
$1+\mathbb{a} + \mathbb{a}^2=0$ and $\mathbb{a}^3=1$.

\begin{definition}[directed linear rank-width]
A {\em directed linear rank decomposition} of digraph $G=(V,E)$ is a pair $(T,f)$, where $T$ is a
caterpillar (i.e.
a path with pendant vertices)
 and $f$ is a bijection between $V$ and the leaves of $T$. 
Each edge $e$ of $T$ divides the vertex set of $G$ by $f$ into two disjoint sets 
$A_e,B_e$. 
For an edge $e$ in $T$ we define
the width of $e$ as $\rg^{(4)}(M_{A_e}^{B_e})$, i.e. the matrix rank\footnote{We denote by $\rg^{(4)}(M)$ the rank of some matrix over 
$\{0,1,\mathbb{a},\mathbb{a}^2 \}$, i.e. the number of independent lines or rows of $M$. 
A set of rows $R$ (i.e. vectors) are {\em independent}, if there is no linear combination 
of a subset $R'$ of $R$ to define a row in $R-R'$.
A {\em linear combination} for some $n$-tuple $r$ is $\sum_{i=1}^n a_ir_i$  for $a_i\in\{0,1,\mathbb{a},\mathbb{a}^2\}$.} of $M$. 
The {\em width} of a directed linear rank decomposition 
$(T,f)$ is the maximal width of all edges in $T$.  
The {\em directed linear rank-width} of a digraph
$G$, $\dlrws(G)$ for short, is the minimum width of all directed linear  rank decompositions for $G$.
\end{definition}

\begin{example}[directed linear  rank-width]\label{ex-lrw} 
\begin{enumerate}[(1.)]
\item
Every bidirectional complete digraph $\overleftrightarrow{K_n}$ and every directed path  $\overrightarrow{P_n}$ has 
directed  linear rank-width 1.

\item Every directed cycle  $\overrightarrow{C_n}$
has directed linear rank-width at most 2.

%\item ... 3

\item Every complete biorientation of a grid $\overleftrightarrow{G_n}$, $n\geq 3$,
has directed   linear rank-width at least $\lceil\frac{2n}{3}\rceil$ and at most $n+1$, see \cite{HOSG08,Gur08b}.

\end{enumerate}
\end{example}

%%%%%%%%%%%%%%%%%%%%%%%%%%%%%%%%%%%%%%%%%%%%%%%%%%%%%%%%%%%%%%%%%%%%%%%%%%%
%%%%%%%%%%%%%%%%%%%%%%%%%%%%%%%%%%%%%%%%%%%%%%%%%%%%%%%%%%%%%%%%%%%%%%%%%%%
\section{Directed width and undirected width}\label{sec-d-u-w}
%%%%%%%%%%%%%%%%%%%%%%%%%%%%%%%%%%%%%%%%%%%%%%%%%%%%%%%%%%%%%%%%%%%%%%%%%%%
%%%%%%%%%%%%%%%%%%%%%%%%%%%%%%%%%%%%%%%%%%%%%%%%%%%%%%%%%%%%%%%%%%%%%%%%%%%

Next we compare the directed width of a digraph $G$
and the undirected width of its  underlying undirected graph
$\un(G)$.

\begin{theorem}\label{th-u-d-w}
Let $G$ be a directed graph. 
\begin{enumerate}[(1.)]
\item \label{th-u-d-w-pw} $\dpws(G)\leq \pws(\un(G))$

\item \label{th-u-d-w-cutw} $\dcutws(G)\leq \cutws(\un(G))$     

\item \label{th-u-d-w-nw} $\nws(\un(G))\leq \dnws(G)\leq \Delta(\un(G)) \cdot \nws(\un(G))$

\item \label{th-u-d-w-nlc}$\lnlcws(\un(G))\leq \dlnlcws(G)\leq \Delta(\un(G)) \cdot \lnlcws(\un(G))+1$

\item \label{th-u-d-w-cw} $\lcws(\un(G))\leq \dlcws(G)\leq \Delta(\un(G)) \cdot \lcws(\un(G))+1$

\item \label{th-u-d-w-rw} $\lrws(\un(G))\leq \dlrws(G) \leq \Delta(\un(G)) \cdot 2^{\lrws(\un(G))+1}-1$
\end{enumerate}
\end{theorem}

\begin{proof}
\begin{enumerate}[(1.)]
\item A path-decomposition for $\un(G)$ of width $k$ is also a directed path-decomposition
for $G$ of width $k$.

\item  Let $G=(V,E)$ be a digraph and $\un(G)$ be the underlying undirected graph
of cut-width $k$. Let $\varphi$ be the corresponding ordering of the
vertices, such that for every $i$, $1\leq i \leq |V|$ there are at
most $k$ edges $\{u,v\}$ such that $u\in L(i,\varphi,\un(G))$ and $v\in R(i,\varphi,\un(G))$.
Since every undirected edge  $\{u,v\}$ in $\un(G)$ comes from a directed edge $(u,v)$,
a directed edge $(v,u)$, or both, and the directed cut-width only counts edges directed
forward, the same layout shows that the directed
cut-width of $G$ is at most $k$.

\item Let $G=(V,E)$ be a digraph of directed neighbourhood-width $k$ and $\varphi\in\Phi(G)$
a linear layout, such that for every $i\in[|V|]$ it holds 
$|N(L(i,\varphi,G),R(i,\varphi,G))|\leq k$. Since for every pair of vertices in $G$
of the same directed neighbourhood the corresponding vertices in $\un(G)$ have the
same neighbourhood, it follows that for every $i\in[|V|]$  it holds 
$|N(L(i,\varphi,\un(G)),R(i,\varphi,\un(G)))|\leq k$. Thus, the 
neighbourhood-width of $\un(G)$ is at most $k$.

Let $G=(V,E)$ be a digraph and $\un(G)=(V,E_u)$ be the underlying undirected graph 
of neighbourhood-width $k$. Then there is 
a layout $\varphi\in \Phi(\un(G))$, such that for every $1\leq i \leq |V|$
the vertices in $L(i,\varphi,\un(G))$ can be divided into at most $k$ subsets 
$L_1,\ldots,L_k$, such that the vertices of set $L_j$, $1\leq j \leq k$, 
have the same neighbourhood with respect to the vertices in $R(i,\varphi,\un(G))$. 
One of these sets $L_j$ may consist of vertices having
no neighbors $v\in R(i,\varphi,\un(G))$.
Every of the remaining sets $L_j$ has at most  $\Delta(\un(G))$ vertices $u$ such that there
is an edge $\{v,u\}\in E_u$ with  $v\in R(i,\varphi,\un(G))$. Let 
$1\leq i \leq |V|$.
\begin{itemize}
\item
If there is one set $L_j$ which consists of vertices having
no neighbours $v\in R(i,\varphi,\un(G))$,
then there are at most $\Delta(\un(G)) \cdot (k-1)$ vertices 
$u \in L(i,\varphi,\un(G))$, such that there
is an edge $\{v,u\}\in E_u$ with  $v\in R(i,\varphi,\un(G))$.

\item 
Otherwise there are at most $\Delta(\un(G)) \cdot k$ vertices 
$u \in L(i,\varphi,\un(G))$, such that there
is an edge $\{v,u\}\in E_u$ with  $v\in R(i,\varphi,\un(G))$. 
\end{itemize}

Thus for every $1\leq i \leq |V|$ 
the vertices in $L(i,\varphi,G)$ can be divided into $k'\leq \Delta(\un(G)) \cdot k$ subsets 
$L'_1,\ldots,L'_{k'}$, such that the vertices of set $L'_j$, $1\leq j \leq k'$, 
have the same directed neighbourhood with respect to the vertices in $R(i,\varphi,G)$.
Thus the directed neighbourhood-width of $G$ is at most $\Delta(\un(G)) \cdot k$.

\item Let $G$ be a digraph of directed linear NLC-width $k$ and $X$ 
be a directed linear NLC-width $k$-expression for $G$. 
A linear NLC-width $k$-expression $c(X)$ for $\un(G)$ can 
recursively be defined as follows.

\begin{itemize}
\item Let $X=\bullet_t$ for $t\in[k]$. Then $c(X)=\bullet_t$.
\item Let $X=\circ_R(X')$ for $R:[k]\to[k]$.  Then $c(X)=\circ_R(c(X'))$.
\item Let $X=X' \otimes_{(\overrightarrow{S},\overleftarrow{S})} \bullet_t$ 
for $\overrightarrow{S}, \overleftarrow{S}\subseteq[k]^2$ and $t\in[k]$.
Then  $c(X)=c(X')\times_{\overrightarrow{S} \cup \overleftarrow{S}} \bullet_t$. 
\end{itemize}

The second bound follows by
$$
\begin{array}{lcl}
\dlnlcws(G) &\stackrel{\text{Lemma \ref{T3b}}}{\leq}&  \dnws(G) +1 \stackrel{\text{(\ref{th-u-d-w-nw}.)}}{\leq} \Delta(\un(G)) \cdot \nws(\un(G)) +1 \\
& \stackrel{\text{\cite{Gur06a}}}{\leq}&
\Delta(\un(G)) \cdot \lnlcws(\un(G))+1.
\end{array}
$$

\item Let $G$ be a digraph of directed linear clique-width $k$ and $X$ 
be a directed linear clique-width $k$-expression for $G$. 
A linear clique-width $k$-expression $c(X)$ for $\un(G)$ can 
recursively be defined as follows.

\begin{itemize}
\item Let $X=\bullet_t$ for $t\in[k]$. Then $c(X)=\bullet_t$.
\item Let $X=X' \oplus  \bullet_t$ for $t\in[k]$. Then $c(X)=c(X')\oplus  \bullet_t$. 
\item Let $X=\rho_{i\to j}(X')$ for $i,j\in [k]$. Then $c(X)=\rho_{i\to j}(c(X'))$.
\item Let $X=\alpha_{i,j}(X')$  for $i,j\in [k]$. Then $c(X)=\eta_{i,j}(c(X'))$.
\end{itemize}

The second bound follows by
$$
\begin{array}{lcl}
\dlcws(G) & \stackrel{\text{Lemma \ref{T3b}}}{\leq}& \dnws(G) +1  \stackrel{\text{(\ref{th-u-d-w-nw}.)}}{\leq} \Delta(\un(G)) \cdot \nws(\un(G)) +1 \\
& \stackrel{\text{\cite{Gur06a}}}{\leq}&
\Delta(\un(G)) \cdot \lcws(\un(G))+1.
\end{array}
$$

\item Let $G$ be a digraph of directed linear rank-width $k$ and $(T,f)$ 
be a directed linear rank-decomposition for $G$ of width $k$. Then $(T,f)$ is also 
a  linear rank-decomposition for $\un(G)$. Let $e$ be an edge of $T$. 
Let $N_{V_1}^{V_2}=(n_{ij})$ be the 
adjacent matrix defined over the two-element field
GF(2) for partition  $V_1\cup V_2$.
If
for $G$ two rows in $M_{A_e}^{B_e}$ are  linearly dependent then 
for $\un(G)$ these two rows in $N_{A_e}^{B_e}$ are also  linearly dependent.
Thus
we conclude that
$\rg^{(2)}(N_{A_e}^{B_e})\leq \rg^{(4)}(M_{A_e}^{B_e})$ and thus
linear rank-width of $\un(G)\leq k$.

The second bound follows by
$$
\begin{array}{lcl}
\dlrws(G) &\stackrel{\text{Lemma \ref{T3ac}}}{\leq}& \dnws(G)  \stackrel{\text{(\ref{th-u-d-w-nw}.)}}{\leq} \Delta(\un(G)) \cdot \nws(\un(G)) \\
& \stackrel{\text{Prop.~6.3 in \cite{OS06}}}{\leq}&
\Delta(\un(G)) \cdot 2^{\lrws(\un(G))+1}-1.
\end{array}
$$
\end{enumerate}
This completes the proof.
\end{proof}

\begin{remark}\label{remark-pw-u}
In Theorem \ref{th-u-d-w}(\ref{th-u-d-w-pw}) and (\ref{th-u-d-w-cutw})
the directed path-width of some digraph  
can not be used to give an upper bound on the path-width of $\un(G)$.
Any transitive tournament has directed path-width $0$ but 
its underlying undirected graph has a path-width which corresponds
to the number of vertices. Also by restricting the vertex degree
this is not possible by an acyclic orientation of a grid.
The same examples also show that the directed cut-width of some digraph  
can not be used to give an upper bound on the cut-width of $\un(G)$.
\end{remark}

The relations shown in Theorem \ref{th-u-d-w} allow to imply the
following values for the directed linear clique-width 
and directed neigh\-bour\-hood-width of a $k$-power graph of a path.

\begin{corollary}\label{power-exact}
\begin{enumerate}[(1.)]
\item For  $n\geq k(k+1)+2$ it holds $\dlcws((\overrightarrow{P_n})^k)= k+2$.
\item For  $n\geq k(k+1)+2$ it holds $\dnws((\overrightarrow{P_n})^k) = k+1$.
\end{enumerate}
\end{corollary}

\begin{proof}
For $n\geq k(k+1)+2$ we know from  \cite{HMP09} that the  
(undirected) linear clique-width of a $k$-power graph of a path
on $n$ vertices is exactly $k+2$. 
\begin{enumerate}[(1.)]
\item 
For  $n\geq k(k+1)+2$ by
$$
k+2 \stackrel{\text{\cite{HMP09}}}{=} \lcws(\un((\overrightarrow{P_n})^k))\stackrel{\text{Theorem \ref{th-u-d-w}}}{\leq} \dlcws((\overrightarrow{P_n})^k)\stackrel{\text{Example \ref{ex-dlcw}}}{\leq} k+2
$$
it holds $\dlcws((\overrightarrow{P_n})^k)= k+2$.

\item 
For  $n\geq k(k+1)+2$ by
$$
\begin{array}{lcl}
k+1 &\stackrel{\text{\cite{HMP09}}}{=} &\lcws(\un((\overrightarrow{P_n})^k))-1\stackrel{\text{Theorem \ref{th-u-d-w}}}{\leq} \dlcws((\overrightarrow{P_n})^k)-1 \\
&\stackrel{\text{Lemma \ref{T3b}}}{\leq} & \dnws((\overrightarrow{P_n})^k)\stackrel{\text{Example \ref{ex-nw}}}{\leq}  k+1
\end{array}
$$
it holds $\dnws((\overrightarrow{P_n})^k) = k+1$.
\end{enumerate}
This completes the proof.
\end{proof}

Comparing the undirected width of a graph $G$ and the directed width 
of its complete biorientation $\overleftrightarrow{G}$ the following
results hold.

\begin{theorem}\label{th-bio} 
For each width measure $\beta\in\{\pws,\cutws,\nws,\lnlcws,\lcws,\lrws\}$ and every
undirected graph $G$ it holds $\beta(G)=d\mbox{-}\beta(\overleftrightarrow{G})$.
\end{theorem}

%\begin{theorem}\label{th-bio} 
%Let $G$ be an undirected graph and $\overleftrightarrow{G}$ its
%complete biorientation. 
%\begin{enumerate}[(1.)]
%
%\item
%$\pws(G)=\dpws(\overleftrightarrow{G})$        
%
%\item 
%$\cutws(G)=\dcutws(\overleftrightarrow{G})$    
%
%\item
%$\lnlcws(G)=\dlnlcws(\overleftrightarrow{G})$
%
%\item
%$\lcws(G)=\dlcws(\overleftrightarrow{G})$
%
%\item
%$\nws(G)=\dnws(\overleftrightarrow{G})$
%
%\item
%$\lrws(G)=\dlrws(\overleftrightarrow{G})$    
%\end{enumerate}
%\end{theorem}

\begin{proof}
\begin{enumerate}[(1.)]
\item  Since $G$ is the underlying undirected graph
of $\overleftrightarrow{G}$, by Theorem \ref{th-u-d-w}(\ref{th-u-d-w-pw}.) it
remains to show that the path-width of $G$ is at most the 
directed path-width of $\overleftrightarrow{G}$. Let $(X_1,\ldots X_r)$
be a directed path-decomposition for  $\overleftrightarrow{G}=(V,E)$.
For every $(u,v)\in E$ it holds $u\in X_i$ and $v\in X_j$ for $i\leq j$.
If $i<j$ then since in $\overleftrightarrow{G}$ there is also the arc $(v,u)$
we obtain a contradiction. Thus it holds $i=j$ which implies that the given
path-decomposition is also a path-decomposition for $G$.

%The result was shown in Lemma 1 of \cite{Bar06}.

\item By Theorem \ref{th-u-d-w}(\ref{th-u-d-w-cutw}.) it
remains to show that the cut-width of $G$ is at most the 
directed cut-width of $\overleftrightarrow{G}$. 
Let $G=(V,E)$ be a graph and $\overleftrightarrow{G}$  its
complete biorientation of directed cut-width $k$. Let $\varphi$ be the corresponding ordering of the
vertices, such that for every $i$, $1\leq i \leq |V|$ there are at
most $k$ arcs $(u,v)$ such that $u\in L(i,\varphi,\un(G))$ and $v\in R(i,\varphi,\un(G))$.
Since every such arc corresponds to one undirected edge $\{u,v\}$ in $G$, 
the same layout shows that the cut-width of $G$ is at most $k$.

\item
By Theorem \ref{th-u-d-w}(\ref{th-u-d-w-nw}.) it
remains to show that the directed neighbourhood-width of $\overleftrightarrow{G}$
is at most the neighbourhood-width of $G$. Let  $\varphi\in\Phi(G)$
a linear layout, such that for every $i\in[|V|]$ for the number of 
neighbourhoods it holds 
$|N(L(i,\varphi,G),R(i,\varphi,G))|\leq k$. By the definitions of $\overleftrightarrow{G}$
and for neighbourhoods of directed graphs, it follows that 
for every $i\in[|V|]$ for the number of directed
neighbourhoods it holds 
$|N(L(i,\varphi,\overleftrightarrow{G}),R(i,\varphi,\overleftrightarrow{G}))|\leq k$.

\item By Theorem \ref{th-u-d-w}(\ref{th-u-d-w-nlc}.) it
remains to show that the directed linear NLC-width of $\overleftrightarrow{G}$
is at most the linear NLC-width of $G$. Let $X$ be an NLC-width $k$-expression
for $G$. A directed NLC-width $k$-expression $c(X)$ for $\overleftrightarrow{G}$
can recursively be defined as follows.
\begin{itemize}
\item Let $X=\bullet_t$ for $t\in[k]$. Then $c(X)=\bullet_t$.
\item Let $X=\circ_R(X')$ for $R:[k]\to[k]$.  Then $c(X)=\circ_R(c(X'))$.
\item Let $X=X' \times_S X''$ for $S\subseteq[k]^2$.
Then  $c(X)=c(X')\otimes_{(S,S)}c(X'')$. 
\end{itemize}

\item
By Theorem \ref{th-u-d-w}(\ref{th-u-d-w-cw}.) it
remains to show that the directed linear clique-width of $\overleftrightarrow{G}$
is at most the linear clique-width of $G$. Let $X$ be a clique-width $k$-expression
for $G$. A directed clique-width $k$-expression $c(X)$ for $\overleftrightarrow{G}$
can recursively be defined as follows.
\begin{itemize}
\item Let $X=\bullet_t$ for $t\in[k]$. Then $c(X)=\bullet_t$.
\item Let $X=X' \oplus X''$. Then  $c(X)=c(X')\oplus c(X'')$. 
\item Let $X=\rho_{i\to j}(X')$ for $i,j\in [k]$.  Then $c(X)=\rho_{i\to j}(c(X'))$.
\item Let $X=\eta_{i,j}(X')$ for $i,j\in [k]$.
Then $c(X)=\alpha_{j,i}(\alpha_{i,j}(c(X')))$.
\end{itemize}

\item
By Theorem \ref{th-u-d-w}(\ref{th-u-d-w-rw}.) it
remains to show that the directed linear rank-width of $\overleftrightarrow{G}$
is at most the linear rank-width of $G$. 
Let $(T,f)$ be a linear rank-decomposition of width $k$ for $G$. 
Then $(T,f)$ is also a  linear rank-decomposition for $\overleftrightarrow{G}$. 
Let $N_{V_1}^{V_2}=(n_{ij})$ be the 
adjacent matrix defined over the two-element field
GF(2) for partition  $V_1\cup V_2$.
Since for every bioriented graph  $N_{V_1}^{V_2}=M_{V_1}^{V_2}$
we conclude that the 
directed linear rank-width of $\overleftrightarrow{G}$ is at most $k$.
\end{enumerate} \nopagebreak
This completes the proof.
\end{proof}

It is already known that recognizing path-width (\cite{ACP87}),  cut-width (\cite{Gav77}),  
linear NLC-width (\cite{Gur06a}), 
linear clique-width (\cite{FRRS09}),  neighbourhood-width (\cite{Gur06a}), and linear
rank-width  (by \cite{Oum17} due \cite{Kas08} and \cite{Oum05a}) are NP-hard.
The results of Theorem \ref{th-bio} imply the same for the directed
versions.

\begin{corollary}\label{np}
Given a digraph $G$ and an integer $k$, then for every 
width measure $\beta\in\{\dpws,\dcutws,\dnws,\dlnlcws,\dlcws,\dlrws\}$, 
the problem to decide whether $\beta(G)\leq k$ is NP-complete.
\end{corollary}

\section{Comparing linear width parameters}

In order to classify graph parameters we use the following notations.
Let ${\mathcal G}$ be the set of all finite directed graphs and
$\alpha,\beta: {\mathcal G}\mapsto \IN$ be two graph parameters.
Parameters $\alpha$ and $\beta$ are called {\em equivalent}, if
there is a function $f_1: \IN\mapsto \IN$ such that for every  $G\in {\mathcal G}$
it holds $\alpha(G)\leq f_1(\beta(G))$ and there is a function $f_2: \IN\mapsto \IN$ such that for every  $G\in {\mathcal G}$
it holds $\beta(G)\leq f_2(\alpha(G))$. 
Parameters $\alpha$ and $\beta$ are called {\em polynomially equivalent}, if
they are equivalent and both functions $f_1$ and $f_2$ are polynomials.
Parameters $\alpha$ and $\beta$ are called {\em linearly equivalent}, if
they are equivalent and both functions $f_1$ and $f_2$ are linear.

%%%%%%%%%%%%%%%%%%%%%%%%%%%%%%%%%%%%%%%%%%%%%%%%%%%%%%%%%%%%%%%%%%%%%%%%%%
\subsection{Relations between linear NLC-width, linear clique-width, neigh\-bour\-hood-width, and linear rank-width}\label{sec-bounds1}
%%%%%%%%%%%%%%%%%%%%%%%%%%%%%%%%%%%%%%%%%%%%%%%%%%%%%%%%%%%%%%%%%%%%%%%%%%

First we state the relation between  the directed linear NLC-width and  
directed linear clique-width. The proofs can be done in the
same way as for the undirected versions  in \cite{GW05a}.

\begin{lemma} \label{T3a}
Let $G$ be a digraph, then it holds
$$
\dlnlcws(G)\leq  \dlcws(G) \leq \dlnlcws(G)+1.
$$
\end{lemma}

Further there is also a very tight connection between the directed neigh\-bour\-hood-width, 
directed
linear NLC-width, and directed linear clique-width. 
The proofs of the following
bounds can be done in a similar fashion as for the undirected versions in \cite{Gur06a}.

\begin{lemma}\label{T3b}
Let $G$ be a digraph, then it holds
$$
\dnws(G)\leq \dlnlcws(G)\leq  \dnws(G)+1
$$
and
$$
\dnws(G)\leq \dlcws(G)\leq  \dnws(G)+1.
$$
\end{lemma}

By the examples given in Section \ref{sec-pre-para} and 
simple observations, we conclude that 
every path $\overrightarrow{P_n}$, $n\ge 3$, has directed linear clique-width 3, 
paths $\overrightarrow{P_3}$ and $\overrightarrow{P_4}$ have 
directed linear NLC-width 2, 
every path $\overrightarrow{P_n}$, $n\ge 5$, has directed linear NLC-width 3, 
and every path $\overrightarrow{P_n}$, $n\ge 3$, has directed neighbourhood-width 2,  
which implies that the bounds of Lemma \ref{T3a} and Lemma \ref{T3b} 
can not be improved.

\begin{lemma} \label{T3ac}
Let $G$ be a digraph, then it holds
$$
\dlrws(G)\leq \dnws(G).
$$
\end{lemma}

\begin{proof}
Let $G$ be a digraph with $n$ vertices of 
directed neighbourhood-width $k$ and $\varphi:V\to [n]$ be a layout such that 
$\dnws(\varphi,G)\leq k$. Using $\varphi$ we define a caterpillar $T_{\varphi}$
with consecutive pendant vertices $\varphi^{-1}(1),\ldots, \varphi^{-1}(n)$.
Pair $(T_{\varphi},\varphi)$ leads to a directed linear rank decomposition for $G$. 
We want to determine the width of $(T_{\varphi},\varphi)$.
Since for every $i$  the vertices in $L(i,\varphi,G)$ define at most 
$k$ neighbourhoods with respect to set $R(i,\varphi,G)$, every edge of
$T_{\varphi}$ leads to a partition of $V$ into $L(i,\varphi,G)$ and $R(i,\varphi,G)$
for some $i$ such that $M_{L(i,\varphi,G)}^{R(i,\varphi,G)}$ has at most $k$ different
rows and thus $\rg(M_{L(i,\varphi,G)}^{R(i,\varphi,G)})\leq k$.
\end{proof}

The following bound can be shown similar to the proof for 
clique-width and rank-width in \cite[Proposition 6.3]{OS06}.

\begin{lemma}\label{L2}
For every digraph $G$ it holds
$$
\dlcws(G)\leq 4^{\dlrws(G)+1} -1.
$$
\end{lemma}

%
% Stimmt das mit 4??
%

The shown bounds imply the following theorem.

\begin{theorem}\label{th-p-w}
Any two parameters in $\{\dnws,\dlnlcws,\dlcws,\dlrws\}$ are equivalent.
\end{theorem}

%\begin{theorem}\label{th-p-w}
%For every class of digraphs $\mathcal G$, the following statements are equivalent.
%\begin{enumerate}
%\item There is some $c$ such that for every digraph $G\in \mathcal G$ it holds $\dlcws(G)\leq c$.
%\item There is some $c$ such that for every digraph $G\in \mathcal G$ it holds $\dlnlcws(G)\leq c$.
%\item There is some $c$ such that for every digraph $G\in \mathcal G$ it holds $\dnws(G)\leq c$.
%\item There is some $c$ such that for every digraph $G\in \mathcal G$ it holds $\dlrws(G)\leq c$.
%\end{enumerate}
%\end{theorem}

%discuss why rankwidth not linear!!!!!!!!!!!!!

We suppose that the exponential bound given in Lemma \ref{L2} can not 
be improved to a linear in general.

\begin{theorem}\label{th-p-w2}
Any two parameters in $\{\dnws,\dlnlcws,\dlcws\}$ are linearly equivalent.
\end{theorem}

Using the arguments of \cite[Section 8]{FOT10} we obtain the next result.

\begin{lemma}\label{lemma-lcw-delta-rw}
There is some polynomial $p$ such that 
for every digraph $G$ it holds
$\dlcws(G)\leq p(\Delta(G),\dlrws(G))$.
\end{lemma}

\begin{theorem}\label{th-delta-lc-lr}
For every class of digraphs  $\mathcal G$ such that for all $G\in\mathcal G$ 
the value $\Delta(G)$ is bounded
any two parameters in $\{\dnws,\dlnlcws,\dlcws,\dlrws\}$ are polynomially equivalent.
\end{theorem}

%%%%%%%%%%%%%%%%%%%%%%%%%%%%%%%%%%%%%%%%%%%%%%%%%%%%%%%%%%%%%%%%%%%%%%%%%%
\subsection{Relations between cut-width and path-width}\label{sec-bounds-c}
%%%%%%%%%%%%%%%%%%%%%%%%%%%%%%%%%%%%%%%%%%%%%%%%%%%%%%%%%%%%%%%%%%%%%%%%%%

The directed path-width is closely related to directed vertex separation number.

\begin{lemma}[\cite{YC08}]\label{th-pw-vs}
For every digraph $G$ it holds
$$
\dpws(G) = \dvsns(G).
$$
\end{lemma}

In \cite{FP13} it is shown how to construct a directed path-decomposition
of width twice the directed cut-width of the graph.
Using the directed vertex separation number, we next show a better bound.

\begin{lemma}\label{th-pw-cutw}
For every digraph $G$ it holds
$$
\dpws(G)\leq \dcutws(G).
$$
\end{lemma}

\begin{proof}
Let $G=(V,E)$ be a digraph of directed cut-width $k$. By (\ref{cutw-d2}) there is 
a layout $\varphi\in \Phi(G)$, such that for every $1\leq i \leq |V|$
there are at most $k$ arcs $(v,u)\in E$  such that 
$v \in R(i,\varphi,G)$ and  $u\in L(i,\varphi,G)$. 
Thus for every $1\leq i \leq |V|$ there are at most $k$ vertices 
$u \in L(i,\varphi,G)$, such that there
is an arc $(v,u)\in E$ with  $v\in R(i,\varphi,G)$.
Thus by (\ref{vsn-d1}) the directed vertex separation number of $G$ is
at most $k$ and by Lemma \ref{th-pw-vs} the directed path-width of $G$
is at most $k$.
\end{proof}

% reverse direction not true in general 

The directed path-width and directed cut-width of a digraph can differ very much, 
e.g.~a $\overleftrightarrow{K_{1,n}}$ has directed path-width 1 and directed 
cut-width $\lceil\frac{n}{2}\rceil$.

\begin{lemma}\label{th-cw-pw2}
For every digraph $G$ it holds
$$
\dcutws(G)\leq  \min(\Delta^{-}(G),\Delta^{+}(G)) \cdot \dpws(G).
$$
\end{lemma}

\begin{proof} 
Let $G=(V,E)$ be a digraph of directed path-width $k$.
By Lemma \ref{th-pw-vs} and (\ref{vsn-d1}) there is 
a layout $\varphi\in \Phi(G)$, such that for every $1\leq i \leq |V|$
there are at most $k$ vertices 
$u \in L(i,\varphi,G)$, such that there
is an arc $(v,u)\in E$ with  $v\in R(i,\varphi,G)$. 
Thus for every $1\leq i \leq |V|$ there are at most $\Delta^{-}(G) \cdot k$ 
arcs $(v,u)\in E$  such that $v \in R(i,\varphi,G)$ and  $u\in L(i,\varphi,G)$. 
By (\ref{cutw-d2}) this implies that the directed cut-width
of digraph $G$ is at most $\Delta^{-}(G) \cdot k$.

The  bound using $\Delta^{+}$ instead of $\Delta^{-}$ can be
shown in the same way  using definition (\ref{vsn-d3a}) instead of (\ref{vsn-d1})
and using definition (\ref{cutw-d1}) instead of (\ref{cutw-d2}).
\end{proof}

%\begin{theorem}\label{th-cutw-pw}
%For every class of digraphs $\mathcal G$, the following statements are equivalent.
%\begin{enumerate}
%\item There is some $c$ such that for every digraph $G\in \mathcal G$ it holds $\dcutws(G)\leq c$.
%\item There is some $c$ and some $d$  such that for every digraph $G\in \mathcal G$ it holds $\dpws(G)\leq c$
%and $\Delta^{-}(G)\leq d$.
%\item There is some $c$ and some $d$  such that for every digraph $G\in \mathcal G$ it holds $\dpws(G)\leq c$
%and $\Delta^{+}(G)\leq d$.
%\end{enumerate}
%\end{theorem}
%
%
%{\bf ?? $\dcutws(G)\geq \frac{\Delta^{-}(G)}{2}$ und $\dcutws(G)\geq  \frac{\Delta^{+}(G)}{2}$
%
%a star with all edges in (out) the center has cw =0}
%
%\begin{proof} 
%$(1)\Rightarrow (2)$ Follows by Theorem \ref{th-pw-cutw} and the fact that $\dcutws(G)\geq  
%\frac{\Delta^{-}(G)}{2}$. $(2)\Rightarrow (1)$ Follows by Theorem \ref{th-cw-pw2}. 
%$(1)\Rightarrow (3)$ by Theorem \ref{th-pw-cutw} and the fact that $\dcutws(G)\geq  
%\frac{\Delta^{+}(G)}{2}$. $(3)\Rightarrow (1)$ Follows by Theorem \ref{th-cw-pw2}.\qed
%\end{proof}

\begin{theorem}\label{th-cutw-pwx}
For every class of  digraphs  $\mathcal G$ such that for all $G\in\mathcal G$ 
the value $\min(\Delta^{-}(G),\Delta^{+}(G))$ is bounded
any two parameters in $\{\dcutws,\dpws\}$ are linearly equivalent.
\end{theorem}

%%%%%%%%%%%%%%%%%%%%%%%%%%%%%%%%%%%%%%%%%%%%%%%%%%%%%%%%%%%%%%%%%%%%%%%%%%
\subsection{Relations between path-width and neighbourhood-width}\label{sec-bounds}
%%%%%%%%%%%%%%%%%%%%%%%%%%%%%%%%%%%%%%%%%%%%%%%%%%%%%%%%%%%%%%%%%%%%%%%%%%

The  directed neighbourhood-width and directed path-width of a digraph can 
differ very much, e.g.~a $\overleftrightarrow{K_{n}}$ has directed neighbourhood-width 1 
and directed path-width $n-1$.

\begin{lemma}\label{th-pw-nw}
For every digraph $G$ it holds
$$
\dpws(G)\leq  \min(\Delta^{-}(G),\Delta^{+}(G)) \cdot \dnws(G).
$$
\end{lemma}

\begin{proof} 
Let $G=(V,E)$ be a digraph of directed neighbourhood-width $k$. Then there is 
a layout $\varphi\in \Phi(G)$, such that for every $1\leq i \leq |V|$
the vertices in $L(i,\varphi,G)$ can be divided into at most $k$ subsets 
$L_1,\ldots,L_k$, such that the vertices of set $L_j$, $1\leq j \leq k$, 
have the same neighbourhood with respect to the vertices in $R(i,\varphi,G)$. Every of these sets
$L_j$ has at most  $\Delta^{-}(G)$ vertices $u$ such that there
is an arc $(v,u)\in E$ with  $v\in R(i,\varphi,G)$.
Thus for every $1\leq i \leq |V|$ there are at most $\Delta^{-}(G) \cdot k$ vertices 
$u \in L(i,\varphi,G)$, such that there
is an arc $(v,u)\in E$ with  $v\in R(i,\varphi,G)$.
Thus by (\ref{vsn-d1}) the directed vertex separation number of $G$ is
at most $\Delta^{-}(G) \cdot k$ and by Lemma \ref{th-pw-vs} the directed path-width of $G$
is at most $\Delta^{-}(G) \cdot k$.

The  bound using $\Delta^{+}$ instead of $\Delta^{-}$ can be
shown in the same way  using  definition (\ref{vsn-d3a}) instead of definition (\ref{vsn-d1}).
\end{proof}

The example 
$\overleftrightarrow{K_{n}}$ shows that the bound  given in Lemma  \ref{th-pw-nw} is tight. 
%
% and also the reverse?!?! wg 2 bounds?!
%

Lemma \ref{th-pw-nw}, Lemma \ref{T3b}, and Lemma \ref{L2}  imply the following bounds.

\begin{corollary}\label{cor-Tpw-deg-n} 
Let $G$ be a digraph, then it holds
$$
\dpws(G)\leq  \min(\Delta^{-}(G),\Delta^{+}(G)) \cdot \dlnlcws(G),
$$
$$
\dpws(G)\leq  \min(\Delta^{-}(G),\Delta^{+}(G)) \cdot \dlcws(G), {\text and}
$$
$$
\dpws(G)\leq  \min(\Delta^{-}(G),\Delta^{+}(G)) \cdot (4^{\dlrws(G)+1} -1).
$$
\end{corollary}

\begin{corollary}\label{cor-pl-or-thres2} 
The directed path-width of a directed threshold graph $G$ 
is at most  $\min(\Delta^{-}(G),\Delta^{+}(G))$.
\end{corollary}

\begin{proof} 
The set of directed threshold graphs  has directed
linear NLC-width 1 (see Theorem \ref{small-lnlcw}). 
Thus the result follows by Corollary \ref{cor-Tpw-deg-n}. 
\end{proof}

Since  $\Delta^-(G)\leq \Delta(G)$ and $\Delta^+(G)\leq \Delta(G)$ and thus 
$$\min(\Delta^{-}(G),\Delta^{+}(G))\leq \Delta(G)$$ the given bounds also 
hold for the more common measure $ \Delta(G)$ instead of $\min(\Delta^{-}(G),\Delta^{+}(G))$.

%\begin{corollary}\label{cor-Tpw-deg} 
%Let $G$ be a digraph, then it holds
%$$
%\dpws(G)\leq   \Delta(G)\cdot\dlnlcws(G).
%$$
%\end{corollary}

%\begin{proof} 
%By the results for undirected graphs in  \cite{Gur06a} we know
%that for every graph $G$ it holds
%$\pws(G)\leq  \Delta(G) \cdot\lnlcws(G)$. By similar arguments as
%in the proof of  Corollary \ref{cor-Tpw-nw2x} the result follows by Theorem \ref{th-u-d-w}(\ref{th-u-d-w-pw})  
%and Theorem \ref{th-u-d-w}(\ref{th-u-d-w-nlc}).\qed
%\end{proof}

After considering the maximum vertex degree, we next make a stronger
restriction by excluding  all possible orientations of a $K_{\ell,\ell}$ as subdigraphs.

\begin{corollary}\label{cor-Tpw-nw2x} 
Let $G$ be a digraph where $\un(G)$ has  no $K_{\ell,\ell}$ subgraph, then it holds
$$
\dpws(G)\leq \pws(\un(G))\leq  2\cdot\lnlcws(\un(G))(\ell-1)\leq  2\cdot\dlnlcws(G)(\ell-1).
$$
\end{corollary}

\begin{proof}
By the results for undirected graphs in  \cite{Gur06a} we know
that for every graph $G$ which has no $K_{\ell,\ell}$ subgraph it holds
$$
\pws(G)\leq  2\cdot\lnlcws(G)(\ell-1).
$$
This implies for every digraph $G$, where $\un(G)$ has  no $K_{\ell,\ell}$ subgraph 
it holds 
$$
\pws(\un(G))\leq  2\cdot\lnlcws(\un(G))(\ell-1).
$$
Furthermore by Theorem \ref{th-u-d-w}(\ref{th-u-d-w-pw})  and Theorem \ref{th-u-d-w}(\ref{th-u-d-w-nlc}) 
for every digraph $G$, where $\un(G)$ has  no $K_{\ell,\ell}$ subgraph 
it holds 
$$
\dpws(G)\leq \pws(\un(G))\leq  2\cdot\lnlcws(\un(G))(\ell-1)\leq  2\cdot\dlnlcws(G)(\ell-1).
$$
This completes the proof.
\end{proof}

\begin{corollary}\label{cor-pl-or-thres} 
Planar directed threshold graphs have directed path-width at most $4$.
\end{corollary}

\begin{proof} 
The set of directed threshold graphs  has directed
linear NLC-width 1 (see Theorem \ref{small-lnlcw}) and for planar digraphs $G$ we know
that  $\un(G)$ has  no $K_{3,3}$ subgraph. 
Thus the result follows by Corollary \ref{cor-Tpw-nw2x}. 
\end{proof}

Next we want to bound the  directed linear clique-width in terms of the
directed path-width.

\begin{remark}\label{rem-pw-lcw}
For general digraphs 
and even for  digraphs of bounded vertex degree
the directed linear clique-width, directed linear NLC-width, directed
neighbour\-hood-width, and directed linear rank-width 
cannot be bounded by the directed path-width by the following examples.
\begin{enumerate}
\item Let $T'$ be an orientation of a tree, e.g.~an out-tree or an in-tree. 
Then $\dpws(T')=0$ by Theorem \ref{small-pw0}. But  $\dlcws(T')$ is unbounded, since
$\lcws(\un(T'))$ is unbounded \cite{GW05a}  
and  since $\lcws(\un(T'))\leq \dlcws(T')$ by Theorem \ref{th-u-d-w}.

%\item Let $T'$ be a biorientation of a tree.
%Then $\dpws(T')=1$ by Theorem \ref{small-pw0}. But  $\dlcws(T')$ is unbounded, since
%$\lcws(\un(T'))$ is unbounded \cite{GW05a}  
%and  since $\lcws(\un(T'))= \dlcws(T')$ by Theorem \ref{th-bio}.

\item Let $G'$ be an acyclic orientation of a grid. 
Then $\dpws(G')=0$ by Theorem \ref{small-pw0}.
But  $\dlcws(G')$ is unbounded, since
$\lcws(\un(G'))$ is unbounded \cite{GR00}  
and  since $\lcws(\un(G'))\leq \dlcws(G')$ by Theorem \ref{th-u-d-w}.

\item
The set of all $k$-power graphs of directed paths has directed path-width $0$
(cf.~Example \ref{ex-dvsn})
and directed linear clique-width $k+2$ (Corollary \ref{power-exact}).
\end{enumerate}
\end{remark}

%This leads to a difference to the relation between the corresponding parameters 
%for undirected graphs (see  \cite{Gur06a}).  

For semicomplete digraphs
the directed path-width can be used to give an upper bound
on the directed clique-width. The main idea of the proof in \cite{FP13}
is to define a directed clique-width expression along 
a nice path-decomposition.\footnote{Please note that in \cite{FP13}
a different notation for directed path-width was used. In Definition \ref{def-dpw}(\ref{def-dpw-2})
the arcs are directed from bags $X_i$ to $X_j$ for $i\leq j$. The authors of \cite{FP13} 
take arcs  from bags $X_i$ to $X_j$ for $i\geq j$ into account. Since an
optimal directed path-decomposition  $(X_1, \ldots, X_r)$ w.r.t.~Definition  \ref{def-dpw}
leads to an optimal directed path-decomposition  $(X_r, \ldots, X_1)$
w.r.t.~the definition of   \cite{FP13}, and vice versa, both definitions lead to
the same value for the directed path-width.}
Since the proof this result in \cite{FP13} only uses 
directed linear clique-width operations we can state the next theorem.

\begin{lemma}[\cite{FP13}]\label{th-cw-pw}
For every semicomplete digraph $S$ it holds
$$
\dlcws(S)\leq  \dpws(S)+2.
$$
\end{lemma}

Lemma \ref{th-cw-pw}, Lemma \ref{T3a}, Lemma \ref{T3b}, and Lemma \ref{T3ac}  imply the following bounds.

\begin{corollary}\label{cor-Tpw-deg-n6} 
For every semicomplete digraph $S$ it holds
$$
\dlnlcws(S)\leq  \dpws(S)+2,
$$
$$
\dnws(S)\leq  \dpws(S)+2, {\text and}
$$
$$
\dlrws(S)\leq  \dpws(S)+2.
$$
\end{corollary}

% linear oder nicht??

\begin{theorem}\label{th-cutw-pw-sem}
For every class of semicomplete digraphs  $\mathcal G$ such that for all $G\in\mathcal G$ 
the value $\min(\Delta^{-}(G),\Delta^{+}(G))$ is bounded
any two parameters in $\{\dpws,\dnws,\dlnlcws,\dlcws,\dlrws\}$ are  equivalent.
\end{theorem}

Using the arguments of \cite[Section 8]{FOT10} and the relations shown in Theorem \ref{th-u-d-w}
there is some polynomial $p$ such that for every digraph $G$ it holds
$$\dpws(G)\leq \pws(\un(G)) \leq p(\Delta(\un(G)),\lrws(\un(G))) \leq p(\Delta(G),\dlrws(G)).$$

\begin{theorem}\label{th-cutw-pw-sem2x}
For every class of semicomplete digraphs  $\mathcal G$ such that for all $G\in\mathcal G$ 
the value $\Delta(G)$ is bounded
any two parameters in $\{\dpws,\dnws,\dlnlcws,\dlcws,\dlrws\}$ are polynomially equivalent.
\end{theorem}

We suppose that the exponential bound given in Corollary \ref{cor-Tpw-deg-n} can not 
be improved to a linear one.

\begin{theorem}\label{th-cutw-pw-sem2}
For every class of semicomplete digraphs  $\mathcal G$ such that for all $G\in\mathcal G$ 
the value $\min(\Delta^{-}(G),\Delta^{+}(G))$ is bounded
any two parameters in $\{\dpws,\dnws,\dlnlcws,\dlcws\}$ are linearly equivalent.
\end{theorem}

By Lemma \ref{th-pw-nw} and Lemma \ref{th-cw-pw}
the restriction to semicomplete digraphs\footnote{When 
considering the directed path-width of almost semicomplete digraphs
in \cite{KKT15} the class
of semicomplete digraphs was suggested to be \gansfuss{a promising stage for pursuing digraph analogues of
the splendid outcomes, direct and indirect, from the Graph Minors project}.} leads to the same
relation between path-with an linear clique-width as for undirected graphs (see  \cite{Gur06a}).

\subsection{Equivalent parameters}

In Table \ref{tab-fam2} we summarize our results on the equivalence
of linear width parameters for directed graphs. For general digraphs
we have three classes of pairwise equivalent parameters, which reduces
to two or one class for $\Delta(G)$ bounded or semicomplete $\Delta(G)$ bounded
digraphs, respectively.

\begin{table}[ht]
\begin{center}
{\tiny
\begin{tabular}{|l|l||c|c|c|c|c|c|c|}
\hline
digraphs                     & equivalence & $\dcutws$ &$\dpws$ &  $\dlcws$ &$\dlnlcws$&$\dnws$&$\dlrws$  \\
\hline
general               &  equivalent                &  \textcolor{LGray}{$\bullet$}&  \textcolor{XGray}{$\bullet$}& \textcolor{BGray}{$\bullet$}& \textcolor{BGray}{$\bullet$}&\textcolor{BGray}{$\bullet$} & \textcolor{BGray}{$\bullet$}\\
                      &  polynomially equivalent   & \textcolor{LGray}{$\bullet$}&  \textcolor{XGray}{$\bullet$}& \textcolor{BGray}{$\bullet$}& \textcolor{BGray}{$\bullet$}&\textcolor{BGray}{$\bullet$} &\\
                      &  linearly equivalent       &\textcolor{LGray}{$\bullet$}&  \textcolor{XGray}{$\bullet$}& \textcolor{BGray}{$\bullet$}& \textcolor{BGray}{$\bullet$}&\textcolor{BGray}{$\bullet$} & \\
\hline
$\Delta(G)$ bounded   &  equivalent                & \textcolor{LGray}{$\bullet$}&  \textcolor{LGray}{$\bullet$}& \textcolor{BGray}{$\bullet$}& \textcolor{BGray}{$\bullet$}&\textcolor{BGray}{$\bullet$} & \textcolor{BGray}{$\bullet$}\\
                      &  polynomially equivalent   & \textcolor{LGray}{$\bullet$}&  \textcolor{LGray}{$\bullet$}& \textcolor{BGray}{$\bullet$}& \textcolor{BGray}{$\bullet$}&\textcolor{BGray}{$\bullet$} & \textcolor{BGray}{$\bullet$}\\
                      &  linearly equivalent       & \textcolor{LGray}{$\bullet$}&  \textcolor{LGray}{$\bullet$}& \textcolor{BGray}{$\bullet$}& \textcolor{BGray}{$\bullet$}&\textcolor{BGray}{$\bullet$} & \\
\hline
semicomplete          &  equivalent                & \textcolor{BGray}{$\bullet$}&  \textcolor{BGray}{$\bullet$}& \textcolor{BGray}{$\bullet$}& \textcolor{BGray}{$\bullet$}&\textcolor{BGray}{$\bullet$} & \textcolor{BGray}{$\bullet$}\\
$\Delta(G)$ bounded   &  polynomially equivalent   & \textcolor{BGray}{$\bullet$}&  \textcolor{BGray}{$\bullet$}& \textcolor{BGray}{$\bullet$}& \textcolor{BGray}{$\bullet$}&\textcolor{BGray}{$\bullet$} & \textcolor{BGray}{$\bullet$}\\
                      &  linearly equivalent       & \textcolor{BGray}{$\bullet$}&  \textcolor{BGray}{$\bullet$}& \textcolor{BGray}{$\bullet$}& \textcolor{BGray}{$\bullet$}&\textcolor{BGray}{$\bullet$} & \\

\hline
\end{tabular}
}
\end{center}
\caption{Classification of linear width parameters for directed graphs. The colors of the points represent
sets of pairwise (linearly, polynomially) equivalent parameters.\label{tab-fam2}}
\end{table}

%%%%%%%%%%%%%%%%%%%%%%%%%%%%%%%%%%%%%%%%%%%%%%%%%%%%%%%%%%%%%%%%%%%%%%%%%%
\section{Characterizations for graphs defined by parameters of small width}\label{sec-ch}
%%%%%%%%%%%%%%%%%%%%%%%%%%%%%%%%%%%%%%%%%%%%%%%%%%%%%%%%%%%%%%%%%%%%%%%%%%

%The classes of graphs of (linear) rank-width 1, tree-width 1 or path-width 1 each pos-
%sess interesting structural properties. For rank-width these are called distance
%hereditary graphs, while for tree-width and path-width we speak of forests and
%disjoint unions of paths respectively.

First we summarize some quite obvious characterizations.

\begin{theorem}\label{small-pw0}
For every digraph $G$ the following statements are equivalent.
\begin{enumerate}[(1.)]
\item $G$ is a DAG.
\item  $\dvsns(G)=0$.
\item  $\dpws(G)=0$.
\item  $\dcutws(G)=0$.
%\item $G$ contains no directed cycle as induced subdigraph.
%\item $G$ has a topological ordering.
\end{enumerate}
\end{theorem}

Next we introduce operations in order to recall the definition of directed 
co-graphs from \cite{BGR97} and introduce an interesting and useful subclass.
Let $G_1=(V_1,E_1)$ and $G_2=(V_2,E_2)$ be two vertex-disjoint directed graphs. 
\begin{itemize}
\item
The {\em disjoint union} of $G_1$ and $G_2$, 
denoted by $G_1 \oplus G_2$, 
is the digraph with vertex set $V_1 \cup V_2$ and 
arc set $E_1\cup E_2$. 

\item
The {\em series composition} of $G_1$ and $G_2$, 
denoted by $G_1\otimes G_2$, 
is the digraph with vertex set $V_1 \cup V_2$ and 
arc set $E_1\cup E_2\cup\{(u,v),(v,u)~|~u\in V_1, v\in V_2\}$. 

\item
The {\em order composition} of $G_1$ and $G_2$, 
denoted by $G_1\oslash  G_2$, 
is the digraph with vertex set $V_1 \cup V_2$ and 
arc set $E_1\cup E_2\cup\{(u,v)~|~u\in V_1, v\in V_2\}$. 
\end{itemize}

\begin{definition}[Directed co-graphs \cite{BGR97}]
The class of directed co-graphs is recursively defined as follows.
\begin{enumerate}[(i)]
\item Every digraph on a single vertex $(\{v\},\emptyset)$, 
denoted by $\bullet$, is a directed co-graph.

\item If $G_1$ and $G_2$ are directed co-graphs, then 
 \begin{inparaenum}[(a)]
\item
$G_1\oplus G_2$,
\item 
$G_1 \otimes G_2$, and
\item
$G_1\oslash G_2$   are directed co-graphs.
\end{inparaenum}
\end{enumerate}
\end{definition}

In \cite{CP06} it has been shown
that directed co-graphs can be characterized by the eight forbidden induced
subdigraphs shown in Table \ref{F-co-ex}.
In \cite{GWY16} the relation of directed co-graphs to the set of
graphs of directed NLC-width 1 and to the set of graphs of directed clique-width 2 
is analyzed.

\begin{table}[ht]
\begin{center}
\begin{tabular}{cccccccc}
\epsfig{figure=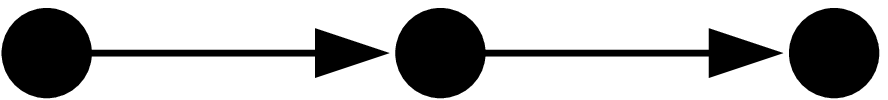,width=2.4cm} &&\epsfig{figure=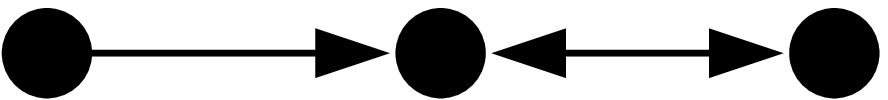,width=2.4cm}&&\epsfig{figure=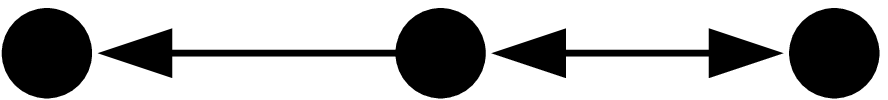,width=2.4cm}&&\\
 $D_1$   &  &    $D_2$   &  &   $D_3$   &  &   \\ 
&&&&&&&\\
\epsfig{figure=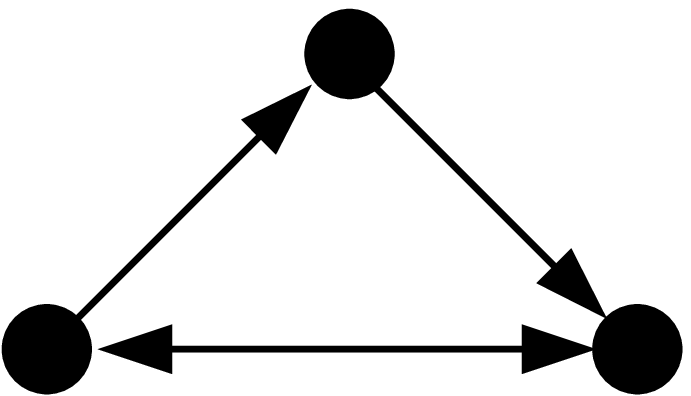,width=2.0cm} &&\epsfig{figure=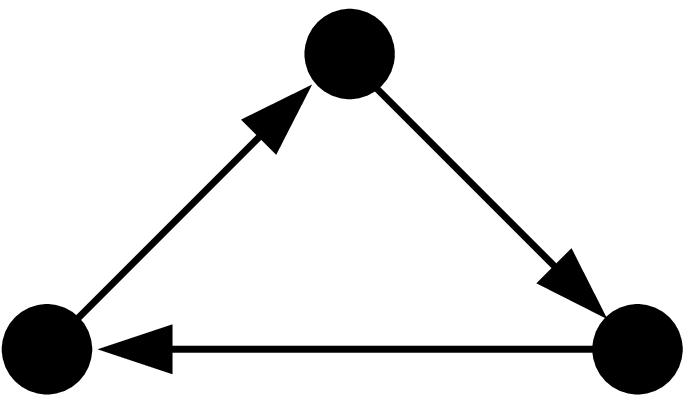,width=1.8cm} &&\epsfig{figure=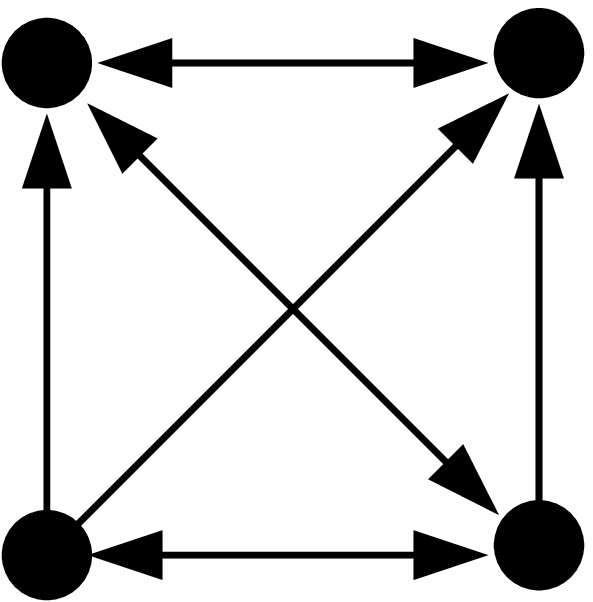,width=1.6cm}&&\\
  $D_4$   &  & $D_5$   &  &    $D_6$   &  &   \\ 
&&&&&&&\\
\epsfig{figure=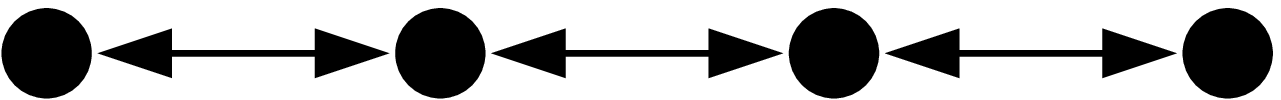,width=3.2cm}&&\epsfig{figure=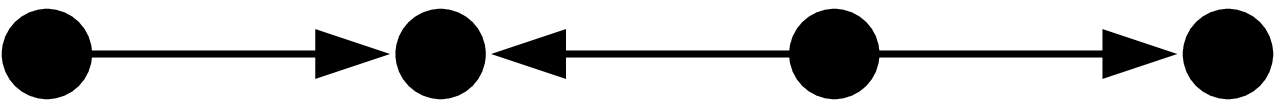,width=3.5cm}&\\
 $D_7$   &  &   $D_8$   & 
\end{tabular}
\end{center}
\caption{The eight forbidden induced subdigraphs for directed co-graphs (see \cite{CP06}).}
\label{F-co-ex}
\end{table}

In order to characterize digraphs of directed linear NLC-width 1 and
digraphs of directed neighbourhood-width 1 we introduce the
following subclass of directed co-graphs.

\begin{definition}[Directed threshold graphs]
The class of directed threshold graphs is recursively defined as follows.
\begin{enumerate}[(i)]
\item Every digraph on a single vertex $(\{v\},\emptyset)$, 
denoted by $\bullet$, is a directed threshold graph.

\item If $G$ is a directed threshold graph,
then 
 \begin{inparaenum}[(a)]
\item
$G\oplus \bullet$,
\item 
$G\oslash \bullet$, 
\item 
$\bullet\oslash G$, and
\item $G\otimes \bullet$ are directed threshold graphs.
\end{inparaenum}
\end{enumerate}
\end{definition}

The related
class oriented threshold graphs was considered by Boeckner in \cite{Boe18}
by using all given operations except the  series composition $G\otimes \bullet$.

\begin{observation}
Every oriented threshold graph is a directed threshold graph and every directed threshold graph
is a directed co-graph.
\end{observation}

\begin{table}[h]
\begin{center}
\begin{tabular}{ccccccccccccc}

\epsfig{figure=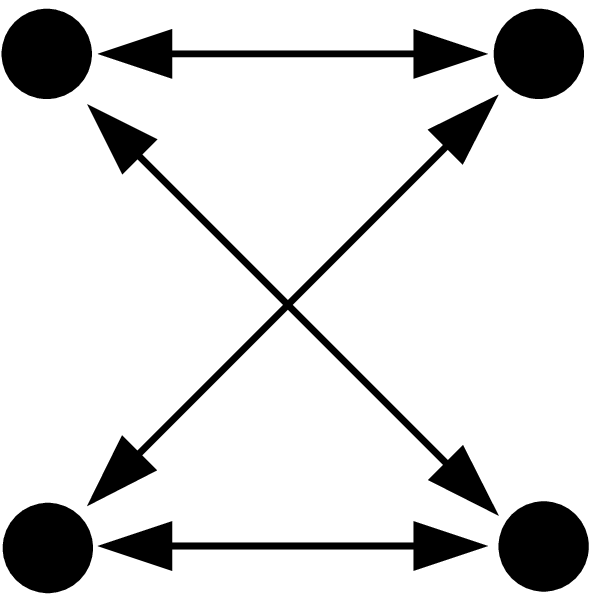,width=1.6cm} &~~&\epsfig{figure=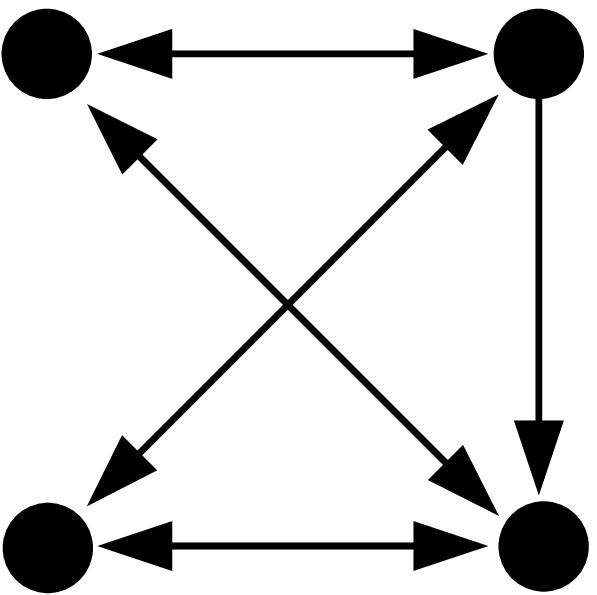,width=1.6cm} &~~& \epsfig{figure=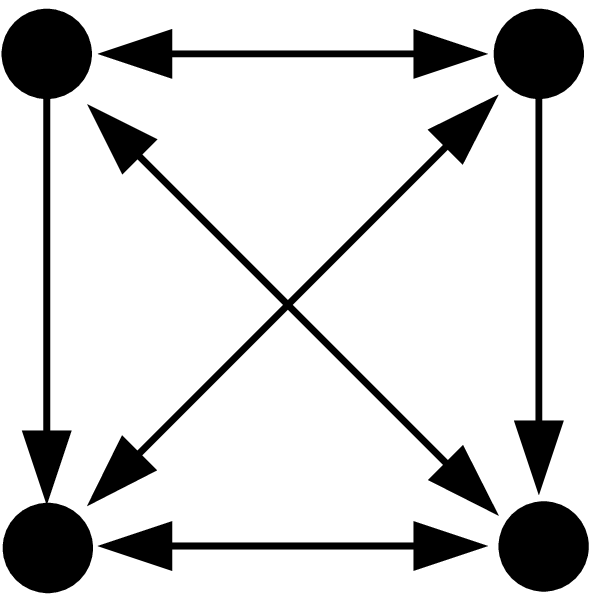,width=1.6cm} &~~&     \\
 $D_9$                             &&    $D_{10}$                        &&    $D_{11}$            &&      \\ 
\end{tabular}
\end{center}
\caption{Forbidden induced subgraphs for subclasses of directed co-graphs.}
\label{F-co21}
\end{table}

\begin{table}[h]
\begin{center}
\begin{tabular}{ccccccccccccc}

\epsfig{figure=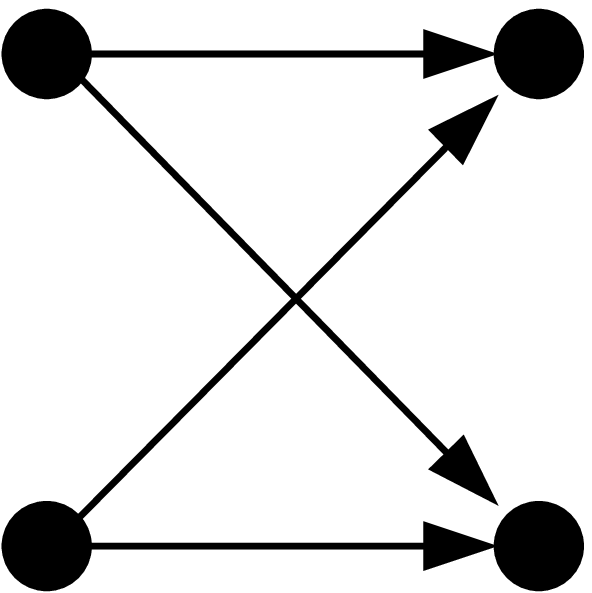,width=1.6cm} &~~&\epsfig{figure=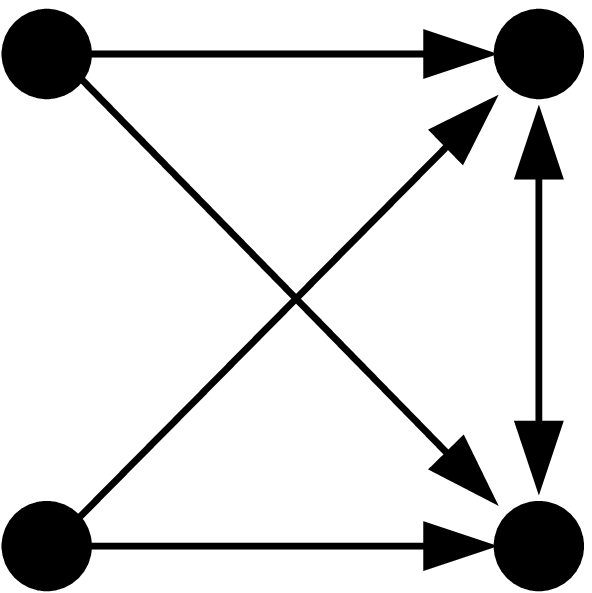,width=1.6cm} &~~& \epsfig{figure=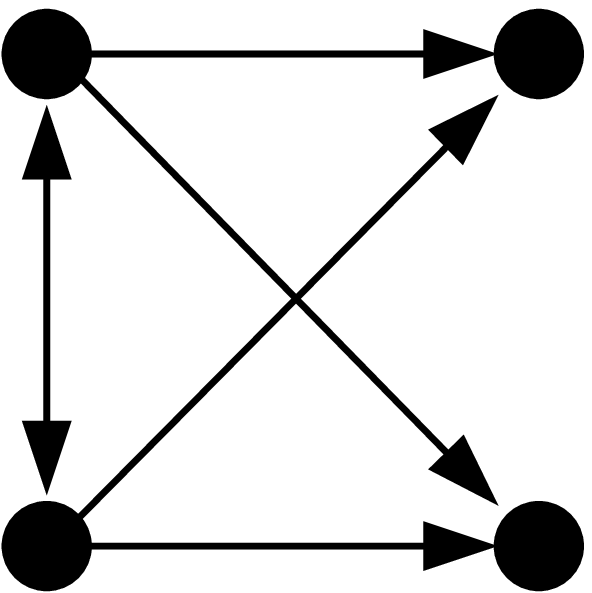,width=1.6cm} &~~&   \epsfig{figure=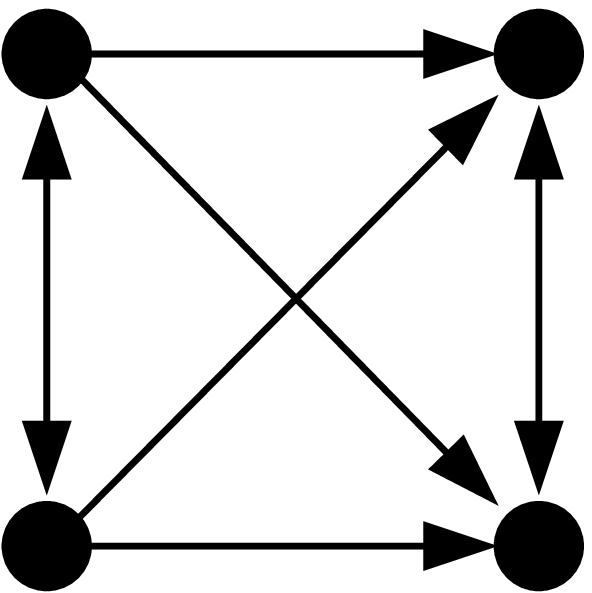,width=1.6cm}  \\
 $D_{12}$                             &&    $D_{13}$                        &&    $D_{14}$            &&   $D_{15}$   \\ 
\end{tabular}
\end{center}
\caption{Forbidden induced subgraphs for subclasses of directed co-graphs.}
\label{F-co2x1}
\end{table}

\begin{table}[ht]
\begin{center}
\begin{tabular}{ccccccccccccc}
\epsfig{figure=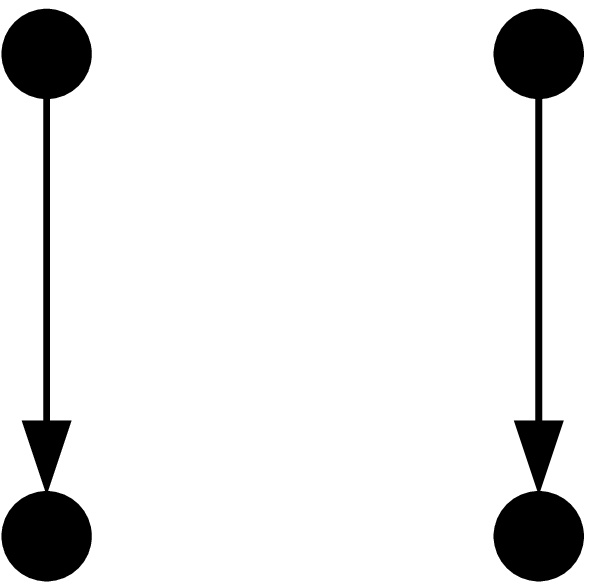,width=1.6cm} &~~&   \epsfig{figure=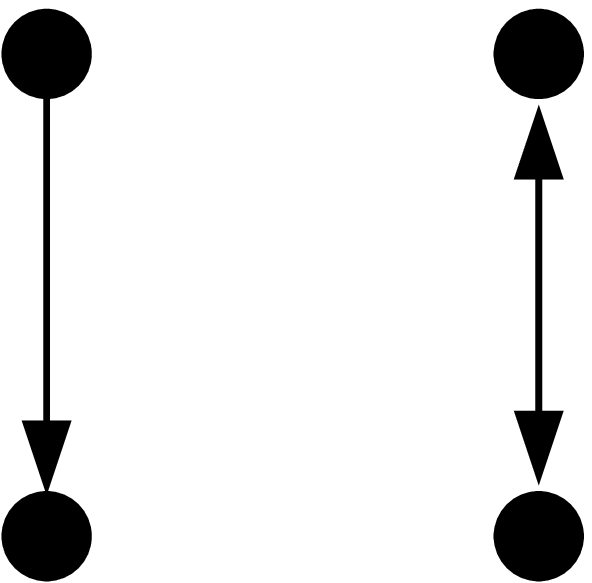,width=1.6cm}                  &~~& \epsfig{figure=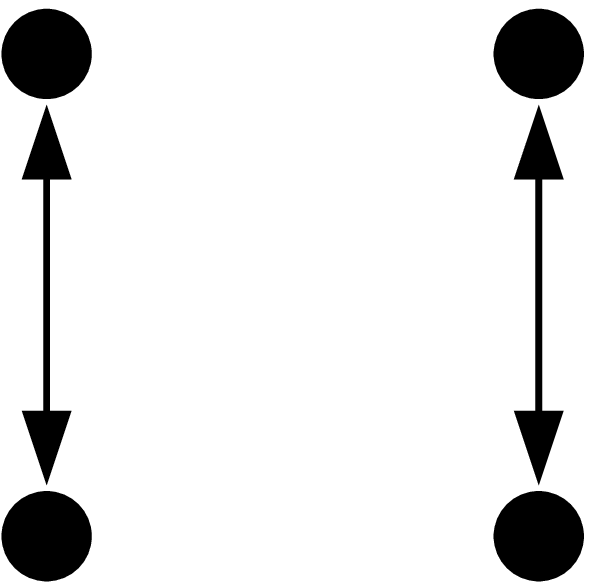,width=1.6cm} \\
 $2\overrightarrow{P_2}$            &  &   $\overrightarrow{P_2}\cup\overleftrightarrow{P_2}$   &  & $2\overleftrightarrow{P_2}$   \\ 
\end{tabular}
\end{center}
\caption{Forbidden induced subgraphs for subclasses of directed co-graphs.}
\label{F-co2a1}
\end{table}

\begin{theorem}\label{small-lnlcw}
For every digraph $G$ the following statements are equivalent.
\begin{enumerate}[(1.)]
\item \label{sm-a} $\dlnlcws(G)= 1$.

\item \label{sm-b} $\dnws(G)= 1$.

\item \label{sm-bn} $\dlcws(G)\leq 2$ and $G\in \free(\{D_2,D_3,D_9,D_{10},
D_{12},D_{13},D_{14}\})$.

\item \label{sm-c} $G$ is a directed threshold graph.

\item \label{sm-d} $G\in \free(\{D_1, \dots, D_{15}, 2\overrightarrow{P_2}, 
		  \overrightarrow{P_2}\cup\overleftrightarrow{P_2},
		  2\overleftrightarrow{P_2}\})$.

\item \label{sm-e} $G\in\free(\{D_1,D_2,D_3,D_4,D_5,D_6,D_{10},D_{11},D_{13},D_{14},D_{15}\})$  
 and $\un(G)\in\free(\{P_4,2K_2,C_4\})$.

\item \label{sm-f} $G\in\free(\{D_1,D_2,D_3,D_4,D_5,D_6,D_{10},D_{11},D_{13},D_{14},D_{15}\})$  
 and $\un(G)$ is a threshold graph.
\end{enumerate}
\end{theorem}

\begin{proof}
$(\ref{sm-a})\Leftrightarrow(\ref{sm-b})$
By the proof of Lemma \ref{T3b} (which can be done similar to the undirected versions in \cite{Gur06a}) 
the set of all digraphs of 
directed linear NLC-width 1 is equal to the set of all
digraphs of directed neighborhood-width 1.

$(\ref{sm-d})\Rightarrow(\ref{sm-c})$ If digraph $G$ does not contain
$D_1,\ldots,D_8$ (see Table \ref{F-co-ex}), then digraph $G$ is a directed co-graph  by \cite{CP06}
and thus has a construction using disjoint union, series composition, and
order composition. 
By excluding $D_9$, $D_{10}$, and $D_{11}$ we know that 
for every series composition of $G_1$ and $G_2$ either $G_1$ 
or $G_2$ is bidirectional complete. Thus this subgraph can also be added by
a number of series operations with one vertex. 

Further by excluding
$D_{12}$, $D_{13}$, $D_{14}$, and $D_{15}$  we know that 
for every order composition of $G_1$ and $G_2$ either $G_1$ 
or $G_2$ is a tournament and since we exclude a directed
cycle of length 3 by $D_5$, we know that $G_1$ 
or $G_2$ is a transitive tournament. Thus this subgraph can also be added by
a number of order operations with one vertex.

By excluding $2\overrightarrow{P_2},\overrightarrow{P_2}\cup\overleftrightarrow{P_2},2\overleftrightarrow{P_2}$
for every disjoint union of $G_1$ and $G_2$ either $G_1$ 
or $G_2$ has no edge. Thus this subgraph can also be added by
a number of disjoint union operations with one vertex.

(\ref{sm-a}) $\Rightarrow$ (\ref{sm-c}):
Let $G=(V,E)$ be a digraph of directed linear
NLC-width 1 and $X$ 
be a directed linear NLC-width $1$-expression for $G$. An expression 
$c(X)$ using directed threshold graph operations for $G$ can recursively 
be defined as follows.

\begin{itemize}
\item Let $X=\bullet_1$ for $t\in[k]$. Then $c(X)=\bullet$.
\item Let $X=\circ_R(X')$ for $R:[1]\to[1]$.  Then $c(X)=c(X')$.
\item Let $X=X' \otimes_{(\overrightarrow{S}, \overleftarrow{S})} \bullet_1$ 
for $\overrightarrow{S}, \overleftarrow{S}\subseteq[1]^2$.

\begin{itemize}
\item If $\overrightarrow{S}=\emptyset$ and $\overleftarrow{S}=\emptyset$,
then $c(X)$ is the disjoint union of $c(X')$ and $\bullet$.

\item If $\overrightarrow{S}=\{(1,1)\}$ and $\overleftarrow{S}=\emptyset$,
then $c(X)$ is the order  composition of $c(X')$ and $\bullet$.

\item If $\overrightarrow{S}=\emptyset$ and $\overleftarrow{S}=\{(1,1)\}$,
then $c(X)$ is the order  composition of $\bullet$ and $c(X')$.

\item If $\overrightarrow{S}=\{(1,1)\}$ and $\overleftarrow{S}=\{(1,1)\}$,
then $c(X)$ is the series  composition of $c(X')$ and $\bullet$.
\end{itemize}
\end{itemize}

(\ref{sm-c}) $\Rightarrow$ (\ref{sm-a}):
Let $G=(V,E)$ be a directed threshold graph and $X$ be an 
expression using directed threshold graph operations for $G$. 
A directed linear NLC-width $1$-expression 
$c(X)$ for $G$ can recursively be defined as follows.

\begin{itemize}
\item
If $X$ defines a single vertex, then $c(X)=\bullet_1$.

\item
If $X$ defines the disjoint union of expression $X_1$ and $\bullet$, then
$c(X)=c(X_1) \otimes_{(\emptyset,\emptyset)}\bullet_1$

\item
If $X$ defines the order composition of expression $X_1$ and $\bullet$, then
$c(X)=c(X_1) \otimes_{(\{(1,1)\},\emptyset)}\bullet_1$

\item
If $X$ defines the order composition of expression of  $\bullet$ and $X_1$, then
$c(X)=c(X_1) \otimes_{(\emptyset,\{(1,1)\})}\bullet_1$

\item
If $X$ defines the series composition of expression $X_1$ and $\bullet$, then
$c(X)=c(X_1) \otimes_{(\{(1,1)\},\{(1,1)\})}\bullet_1$
\end{itemize}

(\ref{sm-c}) $\Rightarrow$ (\ref{sm-bn}): Digraphs $D_2,D_3,D_9,D_{10},D_{12},D_{13},D_{14}$
are not directed threshold graphs. Since directed threshold
graphs are exactly graphs of directed linear NLC-width 1 ((\ref{sm-a}) $\Leftrightarrow$ (\ref{sm-c}))
has been shown above) by Lemma \ref{T3a}
we know that directed threshold graphs have directed linear clique-width at most 2.

(\ref{sm-bn}) $\Rightarrow$ (\ref{sm-d}): Digraphs $D_1,D_4,D_5,D_6,D_7,D_8$
have directed clique-width greater than two and thus directed linear clique-width greater than two. $D_{11},D_{15}$
have  directed linear clique-width at least 3.
Further $2\overrightarrow{P_2},\overrightarrow{P_2}\cup\overleftrightarrow{P_2},2\overleftrightarrow{P_2}$
have an underlying $2K_2$ which has linear clique-width at least 3 and thus by 
Theorem \ref{th-u-d-w}(\ref{th-u-d-w-cw}.) 
the directed linear clique-width of the three digraphs is also at least 3.

$(\ref{sm-c})\Rightarrow(\ref{sm-f})$ If $G$ is a directed threshold graph, then $\un(G)$
is a threshold graph by the recursive definition.
Further the given forbidden digraphs are no directed threshold graphs
and the set of directed threshold graphs is closed under taking induced subdigraphs.

$(\ref{sm-e})\Rightarrow(\ref{sm-d})$ For digraphs $G$ which are excluded within (\ref{sm-d})
but not in (\ref{sm-e}) it holds $\un(G)\in\{P_4,C_4 , 2K_2\}$.

$(\ref{sm-e})\Leftrightarrow (\ref{sm-f})$ Threshold graphs are exactly the set 
$\free(\{P_4,2K_2,C_4\})$ by \cite{CH77}. 
\end{proof}

The set of digraphs of directed linear clique-width 1 is 
exactly the set of edgeless digraphs.
While characterizing digraphs of directed linear NLC-width 1 could be
done by a subclass of directed co-graphs, namely directed threshold graphs, 
this is not possible 
for digraphs of directed linear clique-width 2, since two of the
forbidden induced subdigraphs for directed co-graphs ($D_2$ and
$D_3$) have directed linear clique-width 2.

%%%%%%%%%%%%%%%%%%%%%%%%%%%%%%%%%%%%%%%%%%%%%%%%%%%%%%%%%%%%%%%%%%%%%%%%%%
%%%%%%%%%%%%%%%%%%%%%%%%%%%%%%%%%%%%%%%%%%%%%%%%%%%%%%%%%%%%%%%%%%%%%%%%%%
\section{Conclusions}\label{sec-concl}
%%%%%%%%%%%%%%%%%%%%%%%%%%%%%%%%%%%%%%%%%%%%%%%%%%%%%%%%%%%%%%%%%%%%%%%%%%
%%%%%%%%%%%%%%%%%%%%%%%%%%%%%%%%%%%%%%%%%%%%%%%%%%%%%%%%%%%%%%%%%%%%%%%%%%

We reviewed the  linear clique-width, 
linear NLC-width,
neighbourhood-width, and 
linear rank-width for directed graphs. 
We compared these
parameters with each other and also with the previously 
defined parameters directed path-width 
and directed cut-width. 
While for undirected graphs bounded cut-width implies
directed path-width and directed path-width implies 
directed linear clique-width,
linear NLC-width,
neighbourhood-width, and 
linear rank-width (see \cite{Gur06a}), for directed graphs the relations
turn out to be more involved, see Table  \ref{tab-fam}.
For the restriction to semicomplete  digraphs we obtain the same
relation between these parameters as for undirected graphs (Lemma \ref{th-pw-cutw}, Lemma  \ref{th-cw-pw}, and
Corollary \ref{cor-Tpw-deg-n6}).

With the exception of directed cut-width the considered parameters can be regarded as restrictions
from corresponding parameters with underlying tree-structure to an
underlying path-structure. This implies that the values of the restricted parameters are
greater or equal to the corresponding generalized version.
This relation  can be used to carry over
lower bounds for parameterizations. For example in \cite{FGLSZ18}
parameterizations w.r.t.\ neighborhood-width are used to obtain  parameterizations
w.r.t.\ clique-width.

A further way to define the width of a digraph $G$ is to consider the
width of the underlying undirected graph $\un(G)$. 
This approach is used in works of Courcelle et al.\ \cite{CO00,Cou15a,CE12}, 
when considering the path-width and 
tree-width of directed graphs. We did not follow this 
approach, since it is less sensitive because by using the underlying
undirected graph one can not distinguish the direction of the edges.
%The main idea of tree-width and path-width is to measure how tree-like
%a graph is. 
For example the existence of directed cycles in some digraph $G$ can not be
observed by the path-width of $\un(G)$, while the approach of
Reed, Seymour, and Thomas allows to find directed cycles  
by Theorem \ref{small-pw0}.
%For example every DAG $G$ with a cycle in $\un(G)$ and also
%every directed strong pseudoforest $G$ with a directed cycle in $G$ both have path-width 1
%in this sense. But when  considering the directed path-width every DAG with a cycle in  $\un(G)$ has 
%width 0 and every directed strong pseudoforests with a directed cycle in $G$ has width 1. 
%
%
Using undirected width leads to closer bounds.
If we define for every digraph $G$ the value
$\dpwtwos(G)=\pws(\un(G))$, we
can show the bound given in Lemma \ref{th-cw-pw}
for every arbitrary digraph $G$ as follows.
$$
\dnws(G)\leq  \dpwtwos(G)+1
$$

This implies in connection with Lemma \ref{T3b} the following bounds.
$$
\dlnlcws(G) \leq\dlcws(G) \leq \dpwtwos(G)+2.
$$
The latter bound for the directed linear clique-width in terms of
this approach for directed path-width was also obtained in \cite[Proposition 2.114]{CE12}.

\medskip
There are several interesting open questions. 
\begin{inparaenum}[\bf(a)]
\item By Corollary \ref{np} recognizing all considered  linear width parameters for directed graphs
is NP-hard. There are some xp-algorithms for directed path-width (cf. Section \ref{sec-dpw}) while
xp-algorithms for the other width parameters are uninvestigated up to now.
Further there are fpt-algorithms for computing directed cut-width of semicomplete digraphs \cite{FP13a}
and  computing directed path-width of $\ell$-semicomplete digraphs \cite{KKT15}.
For the other  width parameters fpt-algorithms are unknown, even for semicomplete digraphs.
\item In Theorem \ref{th-u-d-w} 
for several parameters we could show upper and
lower bounds for the directed with of some digraph $G$
in terms of the undirected width of $\un(G)$ and
the maximum vertex degree of $\un(G)$.
For  directed path-width and
directed cut-width this is only possible within one direction 
(see Remark \ref{remark-pw-u}).
It remains to find such bounds for special digraphs, 
e.g.~Eulerian digraphs which were useful for directed
tree-width and tree-width in (2.2) of \cite{JRST01}.
\item In Lemma \ref{L2} we have shown how to bound
the directed linear clique-width of a digraph exponentially in its
directed linear rank-width and in Lemma \ref{lemma-lcw-delta-rw}
we used the results of \cite{FOT10}  to show a polynomial bound
for graphs of bounded vertex degree.
It remains to study the existence of
linear bounds in general and under certain conditions, e.g.~in terms of directed linear rank-width
and further parameters, such as  maximum vertex degree
in order to fill the three open cells in Table \ref{tab-fam2}.
\item In order to characterize digraphs of directed linear rank-width 1 
in terms of special graph operations, we 
propose to generalize 
the notation of thread graphs from \cite{Gan11} to digraphs.
\end{inparaenum}

%%%%%%%%%%%%%%%%%%%%%%%%%%%%%%%%%%%%%%%%%%%%%%%%%%%%%%%%%%%%%%%%%%%%%%
\section{Acknowledgements} \label{sec-a}
%%%%%%%%%%%%%%%%%%%%%%%%%%%%%%%%%%%%%%%%%%%%%%%%%%%%%%%%%%%%%%%%%%%%%%

The work of the second author was supported by the German Research 
Association (DFG) grant  GU 970/7-1.

%%%%%%%%%%%%%%%%%%%%%%%%%%%%%%%%%%%%%%%%%%%%%%%%%%%%%%%%%%%%%%%%%%%%%%%%%%%%%%%%%%%%%%%%%%%%%%%%%%%
%%%%%%%%%%%%%%%%%%%%%%%%%%%%%%%%%%%%%%%%%%%%%%%%%%%%%%%%%%%%%%%%%%%%%%%%%%%%%%%%%%%%%%%%%%%%%%%%%%%

%\bibliographystyle{plain}
%\bibliography{/home/gurski/bib.bib}

\end{document}